\documentclass[11pt,a4paper]{article}

\usepackage[margin=1in]{geometry}
\usepackage{amsmath,amssymb,amsthm}
\usepackage{thmtools, thm-restate}
\usepackage{enumitem}
\usepackage{comment}
\usepackage{graphicx}
\usepackage{xspace}
\usepackage{xcolor}
\usepackage{wrapfig}
\usepackage{makecell}
\usepackage{multirow}
\usepackage{dsfont}
\usepackage{ccfonts,eulervm} \usepackage[T1]{fontenc}\usepackage{hyperref}
\hypersetup{colorlinks=true,linkcolor=black,citecolor=[rgb]{0.5,0,0}}

\usepackage{setspace}
\newcommand{\mS}{\mathcal{S}}
\newcommand{\mT}{\mathcal{T}}
\newcommand{\Q}{\mathcal{Q}}
\newcommand{\R}{\mathcal{R}}

\newcommand{\mP}{\mathcal{P}}
\newcommand{\B}{\mathcal{B}}
\renewcommand{\mP}{\mathcal{P}}
\newcommand{\C}{\mathcal{C}}
\newcommand{\I}{\mathcal{I}}
\newcommand{\D}{\mathcal{D}}

\newcommand{\sq}{S_{\query}}

\newtheorem{theorem}{Theorem}
\newtheorem{definition}{Definition}

\newtheorem{lemma}{Lemma}

\newcommand{\LIS}{\mathsf{LIS}\xspace}

\newcommand{\LISW}{\LIS^{\mathcal{W}}}
\newcommand{\RLIS}{\mathit{Range\mbox{-}LIS}\xspace}
\newcommand{\ORLIS}{\mathit{1D\mbox{-}Range\mbox{-}LIS}\xspace}
\newcommand{\TRLIS}{\mathit{2D\mbox{-}Range\mbox{-}LIS}\xspace}
\newcommand{\CORLIS}{\mathit{Colored\mbox{-}1D\mbox{-}Range\mbox{-}LIS}\xspace}
\newcommand{\CRLIS}{\mathit{Colored\mbox{-}2D\mbox{-}Range\mbox{-}LIS}\xspace}
\newcommand{\WRLIS}{\mathit{Weighted\mbox{-}1D\mbox{-}Range\mbox{-}LIS}\xspace}
\newcommand{\query}{q}
\renewcommand{\tilde}{\widetilde}
\newcommand{\TO}{\widetilde{O}}
\newcommand{\LP}{\mP_{\leq}}
\newcommand{\RP}{\mP_{>}}
\newcommand{\CBMMH}{$\mathsf{CBMMH}$\xspace}
\newcommand{\COVH}{$\mathsf{COVH}$\xspace}

\begin{document}
	\title{Range Longest Increasing Subsequence and its Relatives}


	\author{Karthik C.\ S.\thanks{This work was supported by the National Science Foundation under Grant CCF-2313372 and by the Simons Foundation, Grant Number 825876, Awardee Thu D. Nguyen.}\\  Rutgers University\\ \texttt{karthik.cs@rutgers.edu} \and Saladi Rahul\thanks{This work supported in part by the Walmart Center for Tech Excellence at IISc (CSR Grant WMGT-23-0001).}\\   Indian Institute of Science\\ \texttt{saladi@iisc.ac.in}}

	\date{}
	
	\maketitle
	\thispagestyle{empty}

	\begin{abstract}
		
		Longest increasing subsequence ($\LIS$) is a classical textbook
		problem which is still actively 
		studied in various computational models. In this work, 
		we present a few results for the 
		range longest increasing subsequence 
		problem ($\RLIS$) and its variants. The input to $\RLIS$ is a sequence 
		$\mS$ of $n$ 
		real numbers and  a collection
		$\Q$ of $m$ query ranges and for each query in $\Q$, 
		the goal is to {\em report} the $\LIS$ of the sequence $ \mS $ restricted to that query. Our two main results are for the following generalizations of the $\RLIS$ problem:
		
		\begin{description}
			\item[2D Range Queries:] In this variant of the $\RLIS$ problem, each query is  a pair of ranges, one of indices and the other of values, and we provide a randomized algorithm 
			with running time\footnote{The notation
				$\TO$ hides polylogarithmic factors in $n$ and $m$.}
			$\TO(mn^{1/2}+ n^{3/2})+O(k)$, where $k$ is the 
			cumulative length of the $m$ output subsequences.
			This improves on the elementary 
			$\TO(mn)$ runtime algorithm when $m=\Omega(\sqrt{n})$. 
			Previously, the only known result breaking the quadratic barrier was of Tiskin~[SODA'10] which could only handle 1D range queries (i.e., each query was a range of indices) and also just outputted 
			the {\em length} of the $\LIS$ (instead of reporting the subsequence achieving that length). Subsequent to our paper, Gawrychowski, Gorbachev, and Kociumaka in a preprint have extended Tiskin's approach to handle reporting 1D range queries in $O(n(\log n)^3+m+k)$ time.
			
			\item[Colored Sequences:] In this variant of the $\RLIS$ problem, each element in $\mS$ is colored  and for each query in $\Q$, the goal is to report a monochromatic $\LIS$  contained in the sequence $ \mS $ restricted to that query. For 2D queries, we provide a randomized algorithm for this colored version
			with running time
			$\TO(mn^{2/3}+ n^{5/3})+O(k)$. Moreover, for 1D queries, we provide an improved algorithm with running time
			$\TO(mn^{1/2}+ n^{3/2})+O(k)$.
			Thus, we again improve on the elementary  $\TO(mn)$ runtime algorithm. Additionally, we prove that assuming the well-known Combinatorial Boolean Matrix Multiplication Hypothesis, that the runtime for  1D queries is essentially tight for combinatorial algorithms. 
		\end{description}

		Our algorithms combine several tools
		such as dynamic programming (to precompute
		increasing subsequences with some desirable properties),
		geometric data structures (to efficiently compute the
		dynamic programming entries), random sampling 
		(to capture elements which are part of the $\LIS$), classification 
		of query ranges into large $\LIS$ and small $\LIS$,
		and classification of colors into light and heavy.
		We believe that our techniques 
		will be of
		interest to tackle other variants of $\LIS$ problem 
		and other range-searching problems.

	\end{abstract}
	
	\newpage
	\pagenumbering{arabic}

	\section{Introduction}

	In the {\em longest increasing subsequence} ($\LIS$) problem, the 
	input is a sequence of $n$ real numbers $\mS=(a_1,a_2,\ldots,a_n)$
	and the goal is to {\em report} indices 
	$1\le i_1 < i_2 < \cdots < i_{t}\le n$, for the largest possible value of $t$, such that 
	$a_{i_1} < a_{i_2} < \cdots < a_{i_t}$. The {\em length} of the 
	$\LIS$ is then defined to be $t$.
	For simplicity of discussion, we will assume that all the
	real numbers in the sequence are distinct. A standard 
	dynamic programming algorithm reports the $\LIS$ in $O(n^2)$ time, 
	which can be improved to $O(n\log n)$ by performing 
	patience sort \cite{mallows1962patience,mallows1963patience, a03,fredman1975computing} or by suitably augmenting a binary search tree  
	which computes each entry in the dynamic programming table 
	in $O(\log n)$ time.

	In many applications, such as time-series data analysis, 
	the user might not 
	be interested in the $\LIS$ of the entire sequence.
	Instead they would like to focus their 
	attention only on a ``range'' of indices in the sequence. Here are a few examples to illustrate this point. 
	\begin{itemize}
		\item Consider a stock $X$ in the trading market. 
		Let the sequence represent the  daily 
		price of stock $X$
		from 1950 till 2022. Instead of querying 
		for the $\LIS$ of the entire sequence, the 
		user might gain more insights \cite{jin2007trend,li2017longest} about the 
		trends of the stock by querying
		for the $\LIS$ in a range of dates 
		(or time-windows) such as 
		Jan 2020 till March 2020 or Feb 2018 till Dec 2018. 
		
		\item In modern times, social media platforms 
		are the major source for generating 
		enormous content on the Internet on a 
		minute-to-minute basis. Consider 
		a time-series 
		data with statistics 
		at a fine-grained level (of 
		each minute of the day) about the number 
		of Google searches, number of Tweets posted, 
		number of Amazon purchases, or the number 
		of Instagram posts. Monitoring the 
		$\LIS$ under time-window constraints on such 
		datasets can lead to more insights 
		in understanding the social media usage trends
		and can also help
		in detecting abnormalities.
		
		\item Popular genome sequencing  algorithms identify high scoring segment pairs (HSPs) between a query transcript sequence 
		and a long reference genomic sequence in order to obtain a global
		alignment, and this amounts to computing the $\LIS$ in the reference genome~\cite{altschul1990basic,zhang2003alignment,albert2004longest}.  Typically many queries are made (corresponding to various transcripts/proteins) w.r.t.\ the same reference genomic sequence, and this can be modeled as range queries to compute the $\LIS$. 
	\end{itemize}
	
	Recently, in the theoretical 
	computer science community, there has been a lot of interest in the dynamic $\LIS$ 
	problem~\cite{mitzenmacher2020dynamic, ks21, gj21}, where  
	the goal is to maintain exact or 
	approximate $\LIS$ under insertion and deletion
	of elements. Interestingly, an important subroutine
	which shows up is the 
	computation of 
	$\LIS$ of a given range of elements in 
	the sequence. For example, Gawrychowski and Janczewski~\cite{gj21} 
	design a dynamic algorithm to quickly compute 
	the approximate $\LIS$ for any range of indices in $\mS$. 
	In fact, 
	Gawrychowski and Janczewski go further and design a dynamic algorithm 
	which maintains the approximate $\LIS$ for 
	a special class of {\em dyadic rectangles} where the 
	query\footnote{Consider  mapping  the $i^{\text{th}}$ element $a_i \in \mS$  to a 2D point $(i, a_i)$, to contextualize the notion of 2D queries. } is an axis-aligned rectangle in 2D and the goal is to compute the approximate $\LIS$
	of the 2D points (from the above mapping) which lie inside the rectangle.
	
	Motivated by this, in this paper we 
	study  several natural
	variants of the $\LIS$ problem where the 1D queries impose
	restriction on 
	the range of the indices and the 2D queries impose range restriction
	on both the indices and the  values in the sequence. For each of these variants of the $\LIS$ problem there is a data structure counterpart that can be studied as well, and these results are later listed in Table~\ref{tableDS}.
	
	\paragraph{Vanilla Problem: One dimensional range $\LIS$ problem.}
	
	The first problem is the {\em 1D range longest increasing subsequence} problem 
	($\ORLIS$). In the $\ORLIS$ 
	problem, we are given as input a sequence $\mS=(a_1,a_2,\ldots,a_n)$ 
	of $n$ real numbers and  a collection
	$\Q$ of $m$ query ranges, where each $\query \in \Q$ is a range (or an interval) 
	$[x_1,x_2] \subseteq [1,n]$. For each $\query=[x_1,x_2] \in \Q$, 
	the goal is to {\em report} the $\LIS$ of the sequence $ \mS \cap [x_1,x_2]=(a_{x_1}, a_{x_1 +1}\ldots, a_{x_2-1},a_{x_2})$ (i.e., the output must be the longest increasing subsequence in $ \mS \cap [x_1,x_2]$).
	The length of the $\LIS$ of the sequence $\mS \cap \query$ is denoted by 
	$\LIS(\mS \cap \query)$. Throughout the paper, we  abuse notation for simplicity by using 
	$\LIS$ to refer to both the subsequence and its length. We clarify the intended meaning wherever it may be unclear.
	
	
	A straightforward approach 
	to solve the $\ORLIS$ problem would 
	involve no preprocessing: for each range $[x_1,x_2]$ in $\Q$, we will 
	retrieve the elements
	$a_{x_1}, a_{x_1 +1}\ldots, a_{x_2-1},a_{x_2}$ 
	and report their $\LIS$. This would require
	$\Omega(mn\log n)$ time in the worst-case (for instance when many of the ranges in 
	$\Q$ are of $\Omega(n)$ length). 
	In a remarkable work,
	Tiskin~\cite{Tiskin08a,Tiskin08b,Tiskin10} broke the quadratic barrier, by designing 
	an $O((n+m) \log n)$ time 
	algorithm, albeit for the computationally easier {\em length} version of 
	the 
	$\ORLIS$ problem, in which for all 
	$\query \in \Q$, the goal is to only
	output the length of the $\LIS$ in the range $\query$ (i.e., $\LIS(\mS \cap \query)$).
	While the algorithm designed by Tiskin 
	is tailored to handle the length version of 
	the $\ORLIS$ problem, in a subsequent work to the current paper, Gawrychowski, Gorbachev, and Kociumaka \cite{GGK24} extend Tiskin's algorithm to handle the reporting version as well by providing an algorithm which answers the reporting version of  the 
	$\ORLIS$ problem in $O(n(\log n)^3+m+k)$ time.
	
	Since the vanilla version of the range $\LIS$ problem is  completely understood (up to log factors), we focus on the following three variants.

	\paragraph{Problem-I: Two dimensional range $\LIS$ problem.}
	The next problem is the {\em 2D range longest
		increasing subsequence} problem ($\TRLIS$) which is a 
	natural generalization of the $\ORLIS$ problem.
	As before, we are given as input a sequence
	$\mS=(a_1,a_2,\ldots,a_n)$. 
	Consider a mapping where the $i^{\text{th}}$ element 
	$a_i \in \mS$ is mapped to a 2D point $p_i=(i, a_i)$.
	Let $\mP=\{p_1,p_2,\ldots,p_n\}$ be a collection 
	of these $n$  points. We are also 
	given as input a collection $\Q$ of $m$ query ranges, 
	where each range in $\Q$ is an axis-aligned 
	rectangle in 2D. 
	For each $\query \in \Q$, the goal is to 
	report the $\LIS$ of 
	$\mP \cap \query$, i.e., the points of $\mP$
	lying inside $\query$. See Figure~\ref{fig:2d-rlis-crlis}(a) for an example.
	
	Once again the naive approach to solve the 
	problem would take $O(mn\log n)$ time, 
	where for each range $\query \in \Q$, we 
	will scan the points in $\mP$ to filter
	out the points lying inside $\query$
	and then compute their $\LIS$.
	One might be tempted to use {\em orthogonal 
		range searching} data structures~\cite{aal09,ae98,clp11,nr23,bcko08} to efficiently
	report $\mP \cap \query$, but in the worst-case
	we could have many queries with $|\mP \cap \query|=
	\Omega(n)$.
	In this paper, even for the $\TRLIS$ problem, 
	we succeed to improve on the basic quadratic runtime algorithm
	when $m$ is roughly $\Omega(\sqrt{n})$. 
	We obtain the following result where the $\tilde O(\cdot)$ notation hides polylog factors.
	
	\begin{figure}[h]
		\begin{center}
			\includegraphics[scale=1.2]{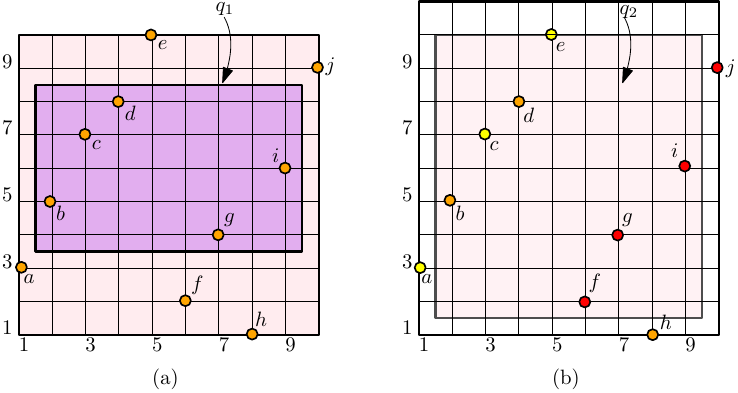}
		\end{center}
		\caption{(a) Each element in the sequence
			$\mS=(3,5,7,8,10,2,4,1,6,9)$ is mapped 
			to a 2D point. For example, the element 
			$7$ is mapped to $c=(3,7)$. The $\TRLIS$ of the points 
			lying inside $\query_1$ is $(b,c,d)$. 
			(b) The same input sequence 
			$\mS=(3,5,7,8,10,2,4,1,6,9)$ is shown 
			here. The set of colors $\C=\{\text{orange, yellow, red}\}$. Points $b,d \text{ and }h$ 
			have orange color, points $a,c\text{ and } e$
			have yellow color, and points 
			$f,g, i \text{ and } j$ have red color.
			The $\CRLIS$ of the points lying inside
			$\query_2$ (the shaded region)
			is red color realized by 
			the sequence $(f,g,i)$.}
		\label{fig:2d-rlis-crlis}
	\end{figure}
	
	\begin{theorem}\label{thm:TRLIS}
		\begin{sloppypar}There is a randomized algorithm to 
			answer the reporting version of 
			$\TRLIS$ problem in 
			$\tilde O(\sqrt{n}\cdot  (m+n))+O(k  )$ time, where $k$ is the cumulative length of the $m$ output subsequences.
			The bound on the running time and 
			the correctness of the solution 
			holds with high probability.\end{sloppypar}
	\end{theorem}
	
	\paragraph{Problem-II: Colored 1D
		range $\LIS$ problem.}
	In the geometric data structures literature, 
	{\em colored range searching}~\cite{kn11,nv13,ptsnv14,m02,lps08,ksv06,bkmt95,agm02,gjs95,jl93,lp12,n14,sj05,gjrs18}
	is a well-studied
	generalization of traditional range searching problems. In the colored setting, the geometric objects are partitioned
	into disjoint groups (or colors), and given 
	a query region $\query$, the typical goal is to efficiently
	report~\cite{chn20,cn20,ch21,lw13}, or count~\cite{krsv07}, or approximately count~\cite{r17}
	the number of colors intersecting  $\query$, or
	find the majority color inside $\query$~\cite{chan2014linear,kms05}.
	
	In the same spirit, we study the {\em colored 1D range 
		longest increasing subsequence 
	} problem ($\CORLIS$). In addition to the sequence 
	$\mS$ and 1D range queries $\Q$ given as input 
	to the $\ORLIS$ problem, in the $\CORLIS$ problem we are
	additionally given a coloring of $\mS$ using the color set $\C$. 
	Each element $a_i \in \mS$ has a {\em color}
	chosen from $\C$ (the corresponding 2D point $p_i$
	also has the same color). For all $c\in \C$, 
	let $\mP_c \subseteq \mP$
	denote the set of points with color $c$ and for any $\query \in \Q$, 
	with a slight abuse of notation, we will 
	use $\LIS(\mP_c \cap \query)$ to denote the length 
	of the $\LIS$ of $\mP_c \cap \query$. Moreover, we say a subsequence is \emph{monochromatic}, if all the elements in the subsequence have the same color. 
	For each $\query \in \Q$, the goal is then to report the longest monochromatic increasing subsequence whose length is equal to  $\underset{c\in C}{\max}\ \LIS(\mP_c \cap \query)$. See Figure~\ref{fig:2d-rlis-crlis}(b) 
	for an example in 2D. We obtain the following 
	result.

	\begin{theorem}\label{thm:CORLIS}
		There is a randomized algorithm to 
		answer the reporting version of the 
		$\CORLIS$ problem in $\tilde O(\sqrt{n}\cdot (m+n))+O(k  )$ time, where $k$ is the cumulative length of the $m$ output subsequences.
		The bound on the running time and 
		the correctness of the solution 
		holds with high probability. 
	\end{theorem}
	
	We complement the above result with a conditional lower bound that indicates that the above runtime for $\CORLIS$ might be (near) optimal. 
	
	\begin{theorem}\label{thm:colorLB}
		Assuming the Combinatorial Boolean Matrix Multiplication Hypothesis, for every $\varepsilon>0$ and $C\in\mathbb{N}$, there is no combinatorial algorithm to 
		answer the reporting version of the 
		$\CORLIS$ problem in $k/2+C\cdot n^{0.5-\varepsilon}\cdot (m+n)$ time, where $k$ is the cumulative size of the $m$ outputs. The lower bound continues to hold  even when we are only required to report the color of the longest monochromatic increasing subsequence for each range query. 
	\end{theorem}
	
	The Combinatorial Boolean Matrix Multiplication Hypothesis (\CBMMH) roughly asserts that Boolean matrix multiplication of two  $\sqrt{n}\times\sqrt{n}$ matrices cannot be computed by any combinatorial algorithm running in time $n^{1.5-\varepsilon}$, for any constant $\varepsilon>0$ (see Definition~\ref{def:CBMMH} for a formal statement). Here ``combinatorial algorithm''  is typically 
	defined as  “non-Strassen-like algorithm” (as defined in
	\cite{ballard2013graph}), and this captures all known fast matrix multiplication algorithms.
	
	While \CBMMH is widely explored since the 1990s \cite{satta1994tree,lee2002fast}, there is no strong consensus in the computer science community about its plausibility. In a recent breakthrough, Abboud et al.\ \cite{abboud2024new}, have come tantalizingly close to even refuting it. All that said, \CBMMH remains a widely used hypothesis for proving conditional lower bounds~\cite{roditty2011dynamic,abboud2014popular,henzinger2015unifying}. At the very least, Theorem~\ref{thm:colorLB} provides evidence that any algorithm for $\CORLIS$ that  is significantly faster than the runtime given in Theorem~\ref{thm:CORLIS}, must be highly non-trivial.
	
	The proof of Theorem~\ref{thm:colorLB} follows as a simple consequence of a \CBMMH based conditional lower bound in \cite{chan2014linear} for computing the mode of a sequence subject to range queries. 
	
	\paragraph{Problem-III: Colored 2D range $\LIS$ problem.}\begin{sloppypar}
		The most general problem that we study in this paper 
		is the {\em colored 2D range 
			longest increasing subsequence 
		} problem ($\CRLIS$). The queries in $\Q$ will be
		axis-aligned rectangles.
		For each $\query \in \Q$, the goal is then to report the longest monochromatic increasing subsequence which attains  $\underset{c\in C}{\max}\ \LIS(\mP_c \cap \query)$.
		See Figure~\ref{fig:2d-rlis-crlis}(b) 
		for an example. We obtain the following result.
	\end{sloppypar}

	\begin{theorem}\label{thm:CRLIS}
		There is a randomized algorithm to 
		answer the reporting version of the 
		$\CRLIS$ problem in $\tilde O(n^{2/3}\cdot (m+n))+O(k  )$ time, where $k$ is the cumulative length of the $m$ output subsequences.
		The bound on the running time and 
		the correctness of the solution 
		holds with high probability. 
	\end{theorem}

	We remark that the conditional lower bound of Theorem~\ref{thm:colorLB} continues to hold here too. 

	In Table~\ref{table}, we have summarized the state-of-the-art upper and (conditional) lower bounds for the various variants of $\RLIS$ discussed in this paper. 
	
	\begin{table}[]
		\begin{center}{\renewcommand{\arraystretch}{1.45}
				\resizebox{0.85\linewidth}{!}{
					\begin{tabular}{||c|c|c||}
						\cline{1-3}
						Problem &  Upper Bound  & Lower Bound   \\ [0.5ex]\hline\hline
						\multirow{2}{*}{1D-Range LIS}  &$O(n(\log n)^{3}+m +k)$ & $\Omega(m+n+k)$  \\
						&(Theorem 5.13 in \cite{GGK24})&(Trivial) \\\hline
						\multirow{2}{*}{2D-Range LIS}  &  $\tilde O(\sqrt{n}\cdot (m+n)) +O(k)$ &  $\Omega(m+n+k)$   \\
						&(Theorem~\ref{thm:TRLIS})&(Trivial) \\\hline
						\multirow{2}{*}{Colored 1D-Range LIS}  & $\tilde O(\sqrt{n}\cdot (m+n)) +O(k)$  &  $\Omega\left(n^{1/2-o(1)}(m+n)+k\right)$  \\
						&(Theorem~\ref{thm:CORLIS})&(Theorem~\ref{thm:colorLB}) \\\hline
						\multirow{2}{*}{Colored 2D-Range LIS}  &   $\tilde O(n^{2/3}\cdot (m+n)) +O(k)$ &  $\Omega\left(n^{1/2-o(1)}(m+n)+k\right)$  \\
						&(Theorem~\ref{thm:CRLIS})&(Theorem~\ref{thm:colorLB}) \\\hline
					\end{tabular}
				}
			}
		\end{center}
		\vspace{-0.1cm}
		\caption{  \footnotesize{State-of-the-art upper and lower bounds are listed for the variants of $\RLIS$ problem considered in this paper. In the table, all upper bounds from this paper are randomized, and all lower bounds from this paper are conditional and only against combinatorial algorithms. Also, recall that $n$ is the length of input sequence, $m$ is the number of range queries, and $k$ is the cumulative size of the $m$ outputs. 
		}}\label{table}
	\end{table}

	\paragraph{Data Structure Counterpart of the Three Variants of $\LIS$ Problem.}
	A natural data structure counterpart question to the aforementioned four problems is to design a data structure for these problems and measure the tradeoff in the time needed for preprocessing the input sequence versus the time needed to answer a single range query. 
	In Table~\ref{tableDS}, we have summarized the preprocessing and query time bounds that can be derived from the proofs of Theorems~\ref{thm:TRLIS},~\ref{thm:CORLIS},~and~\ref{thm:CRLIS}. 
	
	\begin{table}[h]
		\begin{center}{\renewcommand{\arraystretch}{1.45}
				\resizebox{0.85\linewidth}{!}{
					\begin{tabular}{||c|c|c|c||}
						
						\cline{1-4}
						
						Problem &  Preprocessing Time  & Query time &Reference   \\ [0.5ex]\hline\hline
						
						\multirow{1}{*}{1D-Range LIS}  &  $ O(n (\log n)^3 )$ &  $ O(k_q)$  & \cite{GGK24} \\\hline
						
						\multirow{1}{*}{2D-Range LIS}  &  $\tilde O(n^{3/2})$ &  $\tilde O(n^{1/2})+ O(k_q)$  &This Paper \\\hline
						
						\multirow{1}{*}{Colored 1D-Range LIS}  & $\tilde O(n^{3/2})$  &  $\tilde O(n^{1/2}) + O(k_q)$ &This Paper \\\hline
						
						\multirow{1}{*}{Colored 2D-Range LIS}  &   $\tilde O(n^{5/3})$ &  $\tilde O(n^{2/3}) + O(k_q)$  &This Paper\\\hline
						
					\end{tabular}
					
				}
			}
		\end{center}
		\vspace{-0.1cm}
		\caption{  \footnotesize{State-of-the-art upper bounds are listed for the data structure variants of $\RLIS$ problems considered in this paper. In the table, all the bounds are randomized. Also, $k_q$ is the length of the output of a single query.
		}}\label{tableDS}
	\end{table}

	Finally, we note that the proof of Theorem~\ref{thm:colorLB} also implies that assuming \CBMMH, for every $\varepsilon>0$, there is no data structure for the reporting version of the 
	$\CORLIS$ problem (or the $\CRLIS$ problem) which can be preprocessed using a combinatorial algorithm in $O(n^{1.5-\varepsilon})$ time such that each query can be answered using a combinatorial algorithm in $\frac{k_q}{2}+O(n^{0.5-\varepsilon})$ time.

	\subsection{Related Works}
	\paragraph*{$\LIS$ in popular settings and models.}
	To the best of our knowledge, there is not much prior work on the reporting version of $\RLIS$. The results of Tiskin \cite{Tiskin08a, Tiskin08b, Tiskin10} on the length version of $\ORLIS$ were later adapted to the reporting version in \cite{GGK24}, subsequent to our work.
	Therefore, for the rest of this subsection, we list works on the $\LIS$ problem in some popular settings/models. 
	
	In the Dynamic $\LIS$ problem we need to maintain the length of $\LIS$
	of an array under insertion and deletion. The $\LIS$ problem in the dynamic setting was initiated by Mitzenmacher and Seddighin \cite{mitzenmacher2020dynamic}.  In \cite{gj21}, an
	algorithm running in $O(\varepsilon^{-5}(\log n)^{11})$ time  and providing $(1 + \varepsilon)$-approximation was presented (this approximation algorithm can be adapted to approximate the $\RLIS$ problem in the dynamic setting),
	and in \cite{ks21}  a randomized exact algorithm with the update
	and query time $\tilde{O}(n^{2/3})$  was provided.  
	Finally, in \cite{gawrychowski2021conditional}, the authors provide conditional polynomial lower bounds for exactly solving $\RLIS$ in the dynamic setting.

	In the streaming
	model, computing the $\LIS$ requires $\Omega(n)$ bits of space \cite{GopalanJKK07}, so it is natural to resort to
	approximation algorithms.  In \cite{GopalanJKK07}
	is a deterministic $(1 +\varepsilon)$-approximation in $O(\sqrt{n/\varepsilon})$ space for $\LIS$ problem, and this was shown to be optimal \cite{ergun2008distance,gal2010lower}.
	We remark here that the $\LIS$ problem has also been implicitly studied in the streaming algorithms literature as estimating the sortedness of an array \cite{ajtai2002approximate,liben2006finding,GopalanJKK07,sun2007communication}.
	
	In the setting of sublinear time algorithms, the authors of \cite{saks2017estimating} showed how to approximate the length of the $\LIS$ to within an additive
	error of $\varepsilon\cdot n$, for an arbitrary $\varepsilon\in (0, 1)$, with $(1/\varepsilon)^{1/\varepsilon}\cdot (\log n)^{O(1)}$ queries. Denoting the length of $\LIS$
	by $\ell$, in \cite{rubinstein2019approximation} the authors designed a non-adaptive $O(\lambda^{-3})$-multiplicative factor approximation algorithm, where $\lambda=\ell/n$, with
	$O(\sqrt{n}/\lambda^7)$ queries (and also obtained different tradeoff between the dependency on $\lambda$ and $n$). 
	In \cite{newman2021}, the authors proved that adaptivity is essential in obtaining polylogarithmic query
	complexity for the $\LIS$ problem. Recently, for any $\lambda=o(1)$, in \cite{andoni2022estimating}   a (randomized) non-adaptive $1/\lambda^{o(1)}$-multiplicative factor approximation algorithm
	was provided with 
	$n^{o(1)}/\lambda$ running time.
	
	In the read-only random access model, 
	the authors of \cite{kiyomi2018space} showed how to find the length of $\LIS$ in
	$O((n^2/s) \log n)$ time and only $O(s)$ space, for any parameter $\sqrt{n}\le s\le n$.

	\paragraph*{Range-aggregate queries.}
	{\em Range-aggregate queries} have traditionally been well-studied 
	in the database community and the computational geometry community. See the survey 
	papers~\cite{ae98,rt19,gjrs18}. In a typical 
	range-aggregate query, the input is a collection of points in a $d$-dimensional 
	space, and the user specifies a range query (such as an axis-aligned rectangle, or a 
	disk, or a halfspace), and the goal is to compute an informative summary 
	of the subset of points which lie inside the range query, such as 
	the top-$k$ points~\cite{b16,rt15,rt16,abz11,bfgl09}, a random subset of $k$ points~\cite{t22,hqt14}, 
	reporting or counting {\em colors} or groups~\cite{cn20,chn20,r21}, statistics such as 
	mode, min~\cite{bdr10}, median~\cite{bgjs11}, or sum, the closest pair of points~\cite{xue19,xlrj20,xlrj22}, and the skyline points~\cite{bl14,bt11,rj12}. 
	
	In an {\em orthogonal range-max} problem, the input 
	is a set $\mP$ of $n$ points in $d$-dimensional 
	space. Each point in $\mP$ has a weight associated 
	with it. Preprocess $\mP$ into a data structure, 
	so that given an axis-parallel box $\query$, the
	goal is to report the point in $\mP \cap \query$
	with the largest weight. In this work, we will use {\em vanilla range
		trees} from the textbook~\cite{bcko08} to answer
	range-max queries (since the goal of this work is 
	not to shave polylogarithmic factors).
	The preprocessing time to build a vanilla
	range tree is $O(n\log^{d-1}n)$ and the query time
	is $O(\log^dn)$.

	\subsection{Our first technique: Handling small $\LIS$}
	\label{sec:small-lis}
	
	We design a general technique to obtain 
	all our upper bounds. We note that throughout the paper, there is ample scope to shave logarithmic factors in the runtime of our various algorithms and subroutines. However, that is not the focus of this paper.
	
	For the sake of 
	simplicity, we provide an overview of the 
	general technique for the simplest case of 
	$\ORLIS$ which is in 1D and does not involve colors.
	Our strategy to solve the $\ORLIS$ 
	problem is the following. Consider an integer 
	parameter $\tau \in [1,n]$. For the $\ORLIS$ problem,
	$\tau$ is set to
	$\sqrt{n}$. We will design 
	two different techniques. 
	The first technique is deterministic and reports
	the correct answer if $\LIS(\mS \cap \query) 
	\leq \tau$. If $\LIS(\mS \cap \query) > \tau$, then 
	correctness is not guaranteed.
	The second technique
	is randomized and reports the correct answer
	with high probability if $\LIS(\mS \cap \query) 
	\geq \tau$. 
	As such, for each $\query \in \Q$, w.h.p., 
	the correct 
	result is obtained by one of the two algorithms.

	In this subsection, we will give an overview 
	of the first technique for $\ORLIS$.
	In our discussion we will use 
	the terms element and point interchangeably.
	As discussed before, the $i^{\text{th}}$
	element in  sequence $\mS$ is equivalent
	to the 2D point $(i,a_i)$ in $\mP$. 
	Consider the special case where there is a 
	vertical line with $x$-coordinate  $x^*$ such that each query range in $\Q$ 
	contains $x^*$ (we will show later in Section~\ref{sec:first-tec}
	how to remove the assumption).

	\paragraph{Idea-1: Lowest peaks and highest bases.} 
	Please refer to Figure~\ref{fig:peaks-bases}(i)
	where
	eighteen points are shown. Chain $C_1$ 
	and chain $C_2$ are  increasing 
	sequences to the left 
	of $x^*$ and both have length four. 
	Our first observation is that if we 
	are allowed to store either $C_1$ or $C_2$, 
	then it would be wiser to store $C_2$.
	The reason is that $C_2$ can ``stitch''
	itself with chain $C_3$ and $C_4$, 
	to form an increasing subsequence of length 
	six and eight, respectively. On the other hand,
	$C_1$ cannot stitch itself with  $C_3$ or $C_4$ 
	to form a larger increasing subsequence. 
	Analogously, between chains $C_4$ and $C_5$, 
	we will prefer $C_4$ since it can stitch
	with $C_2$, whereas $C_5$ cannot stitch with any
	chain to the left of $x^*$.
	In fact, the $\LIS$ of the eighteen points
	is realised by the chain $C_2 \cup C_4$
	with length eight. This motivates
	the definition of {\em lowest peak} 
	and {\em highest base}.
	
	\begin{figure}[h]
		\begin{center}
			\includegraphics[scale=1]{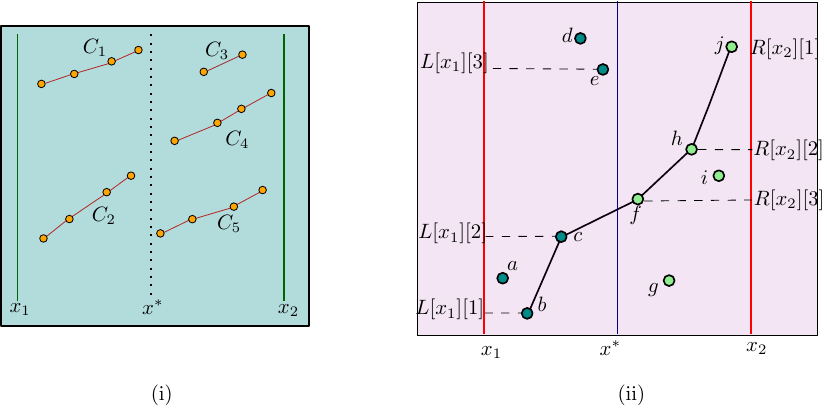}
		\end{center}
		\caption{(i) An example with eighteen points and five chains is shown. (ii) An example with ten points. The
			query is the range $[x_1,x_2]$. In the 
			range $[x_1,x^*]$, the sequence with the 
			lowest peak of length one is $(b)$, of length
			two is $(b,c)$ or $(a,c)$, of length three
			is $(b,c,e)$ or $(a,c,e)$. In the range 
			$[x^*,x_2]$, the sequence with the highest 
			base of length one is $(j)$, of length two
			is $(h,j)$, of length three $(f,h,j)$ or 
			$(f, i,j)$. The $\LIS$ of points in 
			the range $[x_1,x_2]$ is five and 
			the corresponding  compatible pair is
			$(L[x_1][2], R[x_2][3])$.
		}
		\label{fig:peaks-bases}
	\end{figure}
	
	For an increasing subsequence $S \in \mS$, we define {\em peak} 
	(resp., {\em base}) to be the last 
	(resp., first) element in $S$. 
	The algorithm will store two $2$-dimensional arrays:
	\begin{description}
		\item[Array $L$:] For all integers
		$1 \leq x_1 \leq x^*$ and for all 
		$1\leq \alpha \leq \tau$, among 
		all increasing subsequences of length 
		$\alpha$ in the range 
		$[x_1, x^*]$, let $S$ be the subsequence with the {\em lowest} peak.
		Then $L[x_1][\alpha]$ stores the value of the last element in $S$.
		\item[Array $R$:] For all integers $x^* < x_2\leq n$ and for all $1\leq \beta \leq \tau$, among all 
		increasing sequences of length $\beta$ in the range $(x^*,x_2]$,
		let $S$  be the subsequence  with the {\em highest} base. 
		Then $R[x_2][\beta]$ stores the value of the 
		first element in $S$.
	\end{description}
	
	For example, in Figure~\ref{fig:peaks-bases}(i), for the range $[x_1, x^*]$
	and $\alpha=4$, the chain $C_2$
	is the subsequence with the lowest peak; 
	and for $(x^*,x_2]$ and $\alpha=4$, chain $C_4$
	is the subsequence with the highest base.
	Efficient computation of the $O(n\tau)$
	entries in arrays $L$ and $R$ will be discussed later.
	
	A pair $(L[x_1][\alpha], R[x_2][\beta])$ is defined
	to be {\em compatible} if $L[x_1][\alpha] < R[x_2][\beta]$, i.e., 
	it is possible to ``stitch'' the 
	sequences corresponding to $L[x_1][\alpha]$  and 
	$R[x_2][\beta]$ to obtain 
	an increasing sequence of length $\alpha +\beta$.
	For a query range $\query=[x_1,x_2]$, 
	our goal is to find the compatible pair 
	$(L[x_1][\alpha], R[x_2][\beta])$ which 
	maximizes the value of $\alpha + \beta$, i.e., 
	\[\LIS(\mS \cap [x_1,x_2])=\max_{\alpha,\beta}\{\alpha + \beta   \mid 
	L[x_1][\alpha] < R[x_2][\beta]\}.\]
	
	\paragraph{Idea-2: Monotone property of 
		peaks and bases.}
	For a given query $\query=[x_1,x_2]$, 
	a naive approach is to inspect all  
	pairs of the form 
	$(L[x_1][\alpha], R[x_2][\beta])$, 
	for all $1\leq \alpha, \beta \leq \tau$,
	and then pick a compatible pair which maximizes
	the value of $\alpha + \beta$. 
	The number of pairs inspected is $\Theta(\tau^2)$
	which we cannot afford. 
	However, the following {\em monotone property} will help 
	in reducing the number of pairs inspected to 
	$O(\tau\log\tau)$:
	\begin{itemize}
		\item For a fixed value of $x_1$, the value of $L[x_1][\alpha]$ increases 
		as the value of $\alpha$ increases, i.e., 
		for any $1\leq \alpha \leq \alpha' \leq \tau$, we have
		$L[x_1][\alpha] \leq L[x_1][\alpha']$.
		\item Analogously, for a fixed value of $x_2$, the value of $R[x_2][\beta]$ 
		decreases as the value of $\beta$ increases, i.e., 
		for any $1\leq \beta \leq \beta' \leq \tau$, we have 
		$R[x_2][\beta] \geq R[x_2][\beta']$.
	\end{itemize}
	See Figure~\ref{fig:peaks-bases}(ii) for 
	an example  
	where $L[x_1][\alpha]$ values are 
	increasing with increase in $\alpha$
	and $R[x_2][\beta]$ values are 
	decreasing with increase in $\beta$.
	For each $1\leq i \leq \tau$, 
	the largest $\beta_i$ for which 
	$(L[x_1][i], R[x_2][\beta_i])$ is compatible
	can be found in $O(\log \tau)$ time by doing 
	a binary search on 
	the monotone sequence 
	$(R[x_2][1], R[x_2][2],\ldots, R[x_2][\tau])$. 
	Finally,
	report $\max_{i}(i+\beta_i)$ as the length of the 
	$\LIS$ for 
	$\mS \cap \query$. Overall, the time taken to answer
	the $m$ query ranges will be
	$O(m\tau\log \tau)$.

	\paragraph{Idea-3: Pivot points and dynamic programming.}
	Now the key task is to design a preprocessing 
	algorithm which can quickly compute each entry in 
	the two-dimensional arrays $L$ and $R$.
	We will look at $R$ and an analogous discussion will
	hold for $L$. Assume that we have computed the 
	value of $R[x_2-1][\beta]$ and would like to next compute
	$R[x_2][\beta]$. Let $p_{x_2}=(x_2, a_{x_2})$ be the 
	point with $x$-coordinate $x_2$
	which is encountered
	when we move the vertical line from $x_2-1$ to $x_2$. 
	
	Formally, we claim that the value of 
	$R[x_2][\beta]$ is higher than $R[x_2-1][\beta]$ 
	if and only if there is an increasing subsequence 
	with the following three {\em properties}: (i) the subsequence
	length is $\beta$, (ii) the base of the subsequence
	is  higher than  $R[x_2-1][\beta]$, and 
	(iii) the last point in the subsequence is
	$p_{x_2}$. Please refer to Figure~\ref{fig:add-dp}
	for an example.

	\begin{figure}[h]
		\begin{center}
			\includegraphics[scale=1]{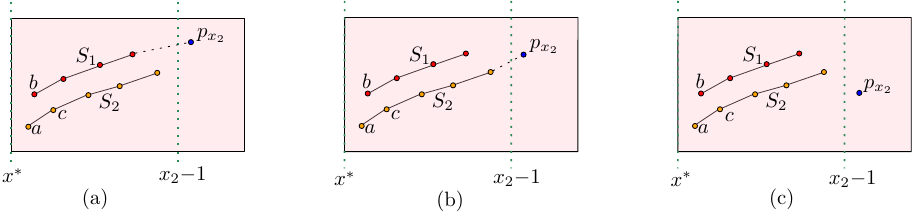}
		\end{center}
		\caption{Two sequences $S_1$ and $S_2$ lie 
			to the left of the vertical line $x_2-1$. The value of $R[x_2-1][5]$ is realized by 
			the $y$-coordinate of point $a$. In (a) and (b), we 
			have $R[x_2][5] > R[x_2-1][5]$, and in (c), we have 
			$R[x_2][5]=R[x_2-1][5]$. In case~(a), point $p_{x_2}$
			stitches with $S_1$ to obtain a subsequence of 
			length five with with higher base than $a$. 
			In case~(b), by stitching 
			point $p_{x_2}$  with $S_2$ and removing point $a$, 
			we obtain a subsequence of length five and it has 
			a higher base than $a$.}
		\label{fig:add-dp}
	\end{figure}

	We define $p_{x_2}$ to be a {\em pivot} point. 
	Now to efficiently compute $R[x_2][\beta]$, 
	we will need the support of another 
	two-dimensional array $B$ defined as follows: among 
	all increasing subsequences of length $\beta$ 
	in the range $(x^*, x_2]$ containing 
	pivot point $p_{x_2}$ as the last element, 
	let $S$ be the subsequence
	with the highest base. 
	Then $B[x_2][\beta]$ is defined to be the value of the first element 
	in $S$. 
	In Figure~\ref{fig:add-dp}(a) and 
	Figure~\ref{fig:add-dp}(b), $B[x_2][5]$ 
	is realized by the $y$-coordinate of $b$ and $c$, 
	respectively.
	As a result, we can succinctly encode
	the three properties stated above as follows:
	\[ R[x_2][\beta] \leftarrow \max\{R[x_2{-}1][\beta], B[x_2][\beta]\}.\] 
	In other words, an increasing sequence of length $\beta$ in the 
	interval $(x^*,x_2]$ with the highest base will either contain the pivot element $p_{x_2}$ or not.
	The first term (i.e., $R[x_2{-}1][\beta]$) captures the case where $p_{x_2}$ is not included, whereas 
	the second term (i.e., $B[x_2][\beta]$) captures the  case where $p_{x_2}$ is 
	included. As such, the task of efficiently computing
	$R$ has been {\em reduced} to the task of efficiently 
	computing $B$.

	\paragraph{Idea-4: 2D range-max data structure.}
	The final step is to compute each entry in 
	$B$ in polylogarithmic time. 
	The entries will be computed in increasing values
	of $\beta$. 
	Assume that the entries 
	corresponding to sequences of length at most $\beta {-}1$ have been computed. Then the entries 
	$B[x_2][\beta]$'s can be reduced to entries of
	$B[\cdot][\beta-1]$'s as follows:
	\begin{align*}
		B[x_2][\beta]=
		\begin{cases}
			a_{x_2}, \quad \text{ if $\beta=1$;} \\
			-\infty, \ \ \ \text{if $a_i > a_{x_2}$ for all $x^*< i < x_2$;} \\
			\max\{B[i][\beta{-}1] \ \vert \ x^*< i < x_2 \text{ and } a_i < a_{x_2}\},
			\quad   \text{otherwise}.
		\end{cases}
	\end{align*}
	We will efficiently compute $B[x_2][\beta]$'s by 
	constructing 
	a data structure which can answer
	{\em 2D range-max queries}. 
	In a 2D range-max query, we are given $n$ weighted points 
	$\mP$ in 2D. Each point is associated with a weight.
	For each point $p_i=(i,a_i)\in \mP$, 
	given a query region of the form 
	$\query'=(-\infty, i) \times (-\infty, a_i)$, 
	the goal is to report the point in $ \mP \cap \query'$
	with the maximum weight.
	The 2D range-max problem can be solved in 
	$O(n\log n)$ time (via reduction to the so-called
	{\em 2D orthogonal 
		point location} problem~\cite{st86}).

	Now with each point $p_i\in \mP$, we 
	associate the weight $B[i][\beta-1]$
	and build a 2D range-max data structure.
	For each point $p_{x_2} =(x_2,a_{x_2}) \in \mP$, the point
	in $\mP \cap 
	(-\infty, x_2) \times (-\infty, a_{x_2})$ with 
	the maximum weight is reported. If point $p_i$ is reported, 
	then set $B[x_2][\beta]=B[i][\beta-1]$.
	The correctness follows from the third case in the above 
	dynamic programming.
	Since $\beta \leq \tau$, the time taken 
	to construct $B$ is
	$O(\tau)\times O(n\log n)=O(n\tau\log n)$.
	
	Therefore, the overall time taken by this 
	technique to answer
	$\ORLIS$ will be $\TO(m\tau + n\tau) + O(k)=\TO(m\sqrt{n} + n\sqrt{n})+O(k)$
	by setting $\tau \leftarrow \sqrt{n}$.
	
	\paragraph{Idea-5: Generalizing $L$ and $R$ to handle 
		$\TRLIS$.} We briefly describe the 
	challenge arising while handling $\TRLIS$ 
	where the queries are axis-parallel rectangles.
	The definition of highest bases and lowest peak
	fail to be directly useful. For example,
	see Figure~\ref{fig:2d-RL}, where $\LIS(\mP \cap \query)$
	is equal to five, and is realized by stitching $S_2$ and $S_3$.
	However, $L[x_1][3]$ will correspond to the sequence $S_1$
	and $R[x_2][2]$ will correspond to the sequence $S_4$, both 
	of which lie completely outside $\query$.
	
	\begin{figure}[h]
		\begin{center}
			\includegraphics[scale=1]{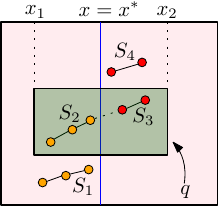}
		\end{center}
		\caption{Four increasing subsequences lying 
			in the range $[x_1,x_2] \times (-\infty,+\infty)$. }
		\label{fig:2d-RL}
	\end{figure}
	
	The key observation in the 2D setting is that for an increasing 
	subsequence $S$ and a query rectangle 
	$\query=[x_1,x_2] \times [y_1,y_2]$, if the first and the last point of $S$ lie inside $\query$, 
	then all the points of $S$ lie inside $\query$.
	Accordingly, we generalize the definitions of 
	$L$ and $R$. Specifically, in the 2D setting,
	$R$ takes two arguments, the query $\query$ and an 
	integer $\beta$, and is defined as follows: 
	\[R(q,\beta) \leftarrow \max_{p_i=(i,a_i)}
	\{ B[i][\beta] \mid p_i\in q, i> x^* \text{ and } B[i][\beta] > y_1 \}.\]
	For example, in Figure~\ref{fig:2d-RL}, 
	the new definition of $R(q,2)$ ensures that the sequence
	$S_3$ gets captured instead of $S_4$. 
	We construct $L$ analogously. Note that 
	in the 1D setting, $L$ and $R$ were {\em explicitly}
	computed. However, in the 2D setting, the
	number of combinatorially different 
	axis-parallel rectangles
	is significantly more than our time 
	budget. Therefore, $L$ and $R$
	{\em cannot} be precomputed and stored. Instead, 
	for all $1\leq \alpha \leq \tau$ (resp.,  
	$1\leq \beta \leq \tau$),
	we will compute $L(\query, \alpha)$ 
	(resp., $R(\query, \beta)$) during the 
	query algorithm in $O(\tau\cdot\text{polylog } n)$ time.

	\subsection{Our second technique: Handling large $\LIS$}
	\begin{sloppypar}
		Now we give an overview of our second technique for the $\ORLIS$ problem (which we show in Theorem~\ref{thm:large-LIS-2d} extends to the $\TRLIS$ problem as well).
		If \mbox{$\LIS(\mS \cap \query) \geq~\tau$}, then this 
		technique will report the answer correctly with 
		high probability.
	\end{sloppypar}
	
	\paragraph{Idea-1: Stitching elements.} 
	One challenge with the $\LIS$ problem is
	that it is {\em not decomposable}. Consider a 
	query range $\query=[x_1,x_2] \times (-\infty,+\infty)$ 
	and an 
	$x^*$ such that $x_1 < x^* < x_2$.
	Let $\query_{\ell}=[x_1,x^*] \times (-\infty,+\infty)$
	and $\query_{r}=(x^*,x_2] \times (-\infty, +\infty)$ 
	such that $q=q_{\ell} \cup q_{r}$.
	Then the $\LIS$ of 
	$\mP \cap \query$ need not be the concatenation of the $\LIS$ of $\mP \cap \query_{\ell}$ and the 
	$\LIS$ of $\mP \cap \query_{r}$,  since the rightmost point
	in the $\LIS$ of $\mP \cap \query_{\ell}$ 
	might have a  larger value 
	than the leftmost point in the $\LIS$ of 
	$\mS \cap \query_{r}$.

	However, suppose
	an oracle reveals that point $p_i\in \mP$ lies in the 
	$\LIS$ of  range $\query=[x_1,x_2] \times 
	(-\infty,+\infty)$. 
	Then we claim that the problem becomes decomposable.
	We will define some notation before proceeding. 
	For a point $p_i=(i,a_i) \in \mP$ corresponding to 
	an element $a_i \in \mS$, define the {\em north-east} region of $p_i$ as $NE(p_i) = \{ (x,y) \mid x \geq i \text{ and } y \geq a_i \}$.
	Analogously, define the {\em north-west}, the 
	{\em south-west} and the {\em south-east} 
	region of $p_i$ which are denoted by 
	$NW(p_i), SW(p_i) \text{ and } SE(p_i)$, 
	respectively.
	Then it is easy to observe that,
	\begin{equation}\label{eq:decomposable}
		\LIS(\mP \cap \query)=\LIS(\mP \cap \query \cap 
		NE(p_i)) 
		+ \LIS(\mP \cap \query \cap SW(p_i)) -1.
	\end{equation}
	See Figure~\ref{fig:stitch}
	for an example. In other words, knowing that $p_i$ belongs to the $\LIS$ decomposes the 
	original $\LIS$ problem into two  $\LIS$ sub-problems which can be 
	computed independently. 
	In such a case, we refer to point $p_i$ as a {\em stitching element},
	which is formally defined below.
	\begin{definition}
		{\bf(Stitching element)} For a range $\query$, 
		fix any $\LIS$ of $\mP \cap \query$ and  call it $S$. 
		Then each element in $S$ is defined to be a stitching element w.r.t. $\query$.
	\end{definition}
	
	The goal of our algorithm is to:
	\begin{center}
		{\em Construct a small-sized set $\R \subseteq \mP$ 
			such that, for any $\query\in \Q$,  at least one 
			stitching element w.r.t. $\query$ is contained in 
			$\R \cap \query$.}
	\end{center}
	
	For each $p_i \in \R$, in the preprocessing 
	phase, the two terms on the right hand side
	of equation~\ref{eq:decomposable} can be computed
	in $O(n\log n)$ time for all possible queries.
	Therefore, the preprocessing time will be
	$O(|\R|n\log n)$.

	\begin{figure*}
		\centering
		\includegraphics{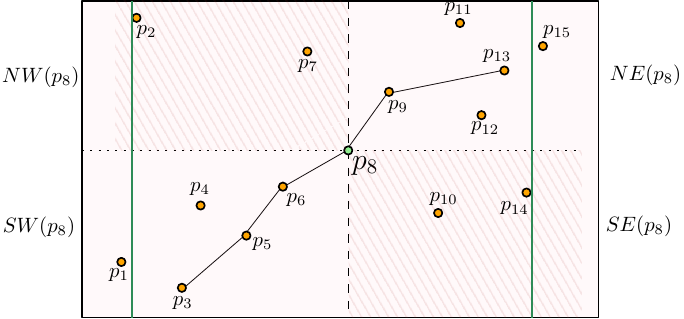}
		\caption{A collection of fifteen elements. 
			Given a query range $\query=[2,14] \times 
			(-\infty,+\infty)$, suppose an oracle reveals
			that $p_8$ lies in the $\LIS$ of $\mP \cap \query$.
			Then observe that we can discard the points 
			in $SE(p_8)$ and $NW(p_8)$ which are shown 
			as shaded regions.  The $\LIS$ of 
			$\mP \cap \query \cap NE(p_8)$ will be 
			$(p_8,p_9, p_{13})$ and and the $\LIS$
			of $\mP \cap \query \cap SW(p_8)$ will be 
			$(p_3,p_5,p_6,p_8)$. Therefore, $\LIS$ of 
			$\mP \cap \query$ will be 
			$(p_3, p_5, p_6, p_8, p_9, p_{13})$.}
		\label{fig:stitch}
	\end{figure*}

	\paragraph{Idea-2: Random sampling.} By using the fact 
	that $\LIS(\mP \cap q) \geq \tau$, it is possible 
	to construct a set $\R$ of size roughly 
	$\frac{n\log n}{\tau} \approx\sqrt{n}\log n$ via random sampling.
	For each $1\leq i\leq n$, sample  $p_i \in \mP$ independently 
	with probability $\frac{c\log n}{\tau}$, where $c$ is a sufficiently 
	large constant. Let $\R \subseteq \mP$ be the set of sampled points. 
	Then we can establish the following 
	connection between $\R$ and the stitching elements.
	
	\begin{restatable}{lemma}{lemlogn}\label{lem:logn-stitching}
		For all ranges $\query$ such that 
		$\LIS(\mP \cap q) \geq \tau$, 
		if $\sq$ is one of the $\LIS$ of $\mP \cap \query$, then 
		with high probability
		$|\R \cap \sq|=\Omega(\log n)$.
	\end{restatable}
	
	This ensures that the preprocessing time 
	will be $O(|\R|n\log n)=O(n^{3/2}\log^2 n)$.
	To answer any range $\query$, first scan $\R$ 
	to identify the points
	which lie inside $\query$. 
	For each element $p_i \in \R \cap \query$, compute 
	the right hand side of equation~\ref{eq:decomposable} 
	in $O(\log n)$ time. 
	Finally, report the largest value computed.
	Therefore, the query time is bounded by 
	$O(|\Q|\cdot |\R|\cdot \log n)=O(m\sqrt{n}\log^2n)$.

	\subsection{Our third technique: 
		Handling small cardinality colors}
	
	\paragraph{A naive approach.} The $\CRLIS$ problem is  more 
	challenging than the $\TRLIS$ problem. Firstly, 
	the result of Theorem~\ref{thm:TRLIS} does not 
	really help in answering $\CRLIS$ since it 
	completely ignores the information about the 
	colors. Secondly, there might be a temptation
	to solve the problem ``independently'' 
	for each color class. For example, consider 
	a setting where each color class has roughly 
	equal number of points, i.e., 
	$n/|\C|$ points. Consider a color $c$
	and build an instance of Theorem~\ref{thm:TRLIS} 
	based on the points of color $c$. Then, for each 
	query $\query \in \Q$, we have the value of 
	$\LIS(\mP_c \cap q)$. Repeat this for each color
	in $\C$. Finally, for each query $\query \in \Q$, 
	report the color $c$ for which $\LIS(\mP_c \cap q)$
	is maximized.
	
	Now lets (informally) analyze this algorithm.
	The running time bound of Theorem~\ref{thm:TRLIS}
	has three terms. For now, lets ignore the second term 
	and the third term, which represent the time
	taken to build the data structure and the time 
	taken to report the output of the $m$ queries.
	The first term is the time taken to answer 
	$m$ queries (which is roughly $\sqrt{n}$ time per query).
	Adding the first term for each color class, we get
	\[ \sum_{c} \TO(m|\mP_c|^{1/2}) \leq \TO(|\C| m(n/|\C|)^{1/2})
	=\TO(m\sqrt{n|\C|}),\]
	which is $\TO(mn)$ when $|\C|=\Theta(n)$
	and $\TO(mn^{\frac{1+\varepsilon}{2}})\gg m\sqrt{n}$ when 
	$|\C|=n^{\varepsilon}$, for any $0 < \varepsilon \leq 1$.
	Therefore, in the worst-case this approach does not help
	achieve our target query time bound.

	\paragraph{A technique to handle small cardinality 
		colors.}
	We will design a third technique which will be
	helpful to answer $\CORLIS$ and $\CRLIS$.
	The key idea is to handle all the colors in a 
	``combined'' manner. The technique will work 
	well when all the colors have small cardinality.
	Given a parameter $\Delta$, assume that for each color $c$, 
	the value of $|\mP|_c \leq \Delta$. 
	The key observations made by the algorithm are 
	the following:
	\begin{itemize}
		\item {\em Bounding the number of output 
			subsequences.} For a query $\query$, let
		the output $\LIS$ be $S$ of color $c$. Let 
		$p_i \in \mP_c$ (resp., $p_j \in \mP_c$) be the 
		first (resp., last) point on $S$. Call this a pair $(i,j)$.
		As such, whenever color 
		$c$ is the output, the number of distinct pairs
		will be $O(|\mP_c|^2)$. Adding up 
		over all the colors, the total 
		number of distinct pairs will be:
		\[\sum_{c}O(|\mP_c|^2) \leq O\left(\Delta\sum_{c}|\mP_c|\right)=O(n\Delta).\]
		If $\Delta \ll n$, then this is a significant
		improvement over the naive bound of $O(n^2)$ 
		distinct pairs (or distinct output subsequences).
		\item {\em Reduction to an uncolored problem.}
		For each possible output subsequence $S$ of 
		color $c$ with $p_i \in \mP_c$ (resp., $p_j \in \mP_c$)
		as the first (resp., last) point point on $S$, 
		we construct an axis-aligned rectangle 
		$R_c(i,j)$ with $p_i$ (resp., $p_j$) as the
		bottom-left (resp., top-right) corner. 
		A weight $w(R_c(i,j))$ is associated with rectangle
		$R_c(i,j)$ which is equal to $\LIS(\mP_c \cap R_c(i,j))$.
		See Figure~\ref{fig:bounding-box}.
		Let $\R$ be the collection of $O(n\Delta)$ such weighted
		rectangles. As a result, the problem of 
		$\CRLIS$ is now reduced to the 
		{\em rectangle range-max problem}, where given
		an axis-parallel rectangle $\query$, among all the rectangles
		in $\R$ which lie completely inside $\query$, 
		the goal is to report 
		the rectangle  with the largest weight. 
		
		It is crucial to note that the colored 
		problem has been reduced to a problem on
		{\em uncolored} rectangles for which an 
		efficient data structure exists: the 
		rectangle range-max data structure 
		can be constructed in $\TO(|\R|)=\TO(n\Delta)$ time
		and the $m$ range queries can be answered in 
		$\TO(m) + O(k)$ time.
	\end{itemize}
	
	\begin{figure}
		\centering
		\includegraphics{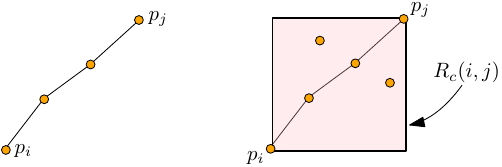}
		\caption{On the left is an increasing subsequence 
			of length four with all four points
			being of color $c$. On the right is the corresponding
			rectangle $R_c(i,j)$. The $\LIS$ of color $c$ inside
			$R_c(i,j)$ is four (the other two points of color 
			$c$ do not participate in the $\LIS$). Therefore,
			the weight associated with $R_c(i,j)$ is four.}
		\label{fig:bounding-box}
	\end{figure}

	\subsection{Putting all the pieces together}
	As an illustration we will consider the 
	$\CRLIS$ problem and put together all the 
	three techniques discussed in the previous subsections in a specific manner.
	
	\paragraph{Light and heavy colors.}
	Define a parameter $\Delta$ which will be set later.
	A color $c$ is classified as {\em light} if 
	$|\mP_c|\leq \Delta$, otherwise, it is classified
	as {\em heavy}. We will design different algorithms
	to handle light colors and heavy colors.
	The advantage with a light color, say $c$, is that
	we can precompute the $\LIS$ for $O(|\mP_c|^2)$ 
	2D axis-parallel 
	ranges and still be within the time budget. 
	On the other hand, the advantage with heavy colors
	is that there can be only $O(n/\Delta)$ heavy colors.
	
	\paragraph{Handling light colors.}
	Let $\mP_{\ell} \subseteq \mP$ 
	be the set of points which belong to a light color.
	We will use the third technique to answer 
	$\CRLIS$ on $\mP_{\ell}$
	and queries $\Q$.
	From Section~\ref{sec:third-tec}, it follows that 
	the running time is $\TO(m+n\Delta) +O(k)$.

	\paragraph{Handling heavy colors and small $\LIS$ 
		queries.}
	Let $\mP_{h} \subseteq \mP$ 
	be the set of points which belong to a heavy color.
	In Section~\ref{sec:first-tec}, we will prove
	that for an arbitrary value of $\tau$, the running time
	of the first technique will be $\TO(m\tau + n\tau)$.
	The first technique as described above only 
	handles uncolored points. In Section~\ref{sec:first-tec},
	we will ``generalize'' the first technique to
	answer $\CRLIS$ problem as well. 
	The running time of this generalized first technique 
	on colored pointset $\mP_{h}$ and queries
	$\Q$ will be 
	$\TO \left(\frac{m\tau n}{\Delta} + \frac{n^2\tau}{\Delta}\right) +O(k)$, where the number of heavy colors is $O(n/\Delta)$.
	
	\paragraph{Handling heavy colors and large $\LIS$ 
		queries.}
	In Section~\ref{sec:second-tec}, we will prove that
	for an arbitrary $\tau$, the running time of the 
	second technique will be $\TO \left( \frac{mn}{\tau}+\frac{n^2}{\tau}\right) + O(k)$ for $\TRLIS$ problem.
	In fact, we will generalize this algorithm 
	to the large $\LIS$ case of $\CRLIS$ with the 
	same running time (ignoring polylogarithmic factors).
	The generalized algorithm will answer 
	$\CRLIS$ on $\mP_h$ and queries 
	$\Q$.

	Combining all the three subroutines, the total running
	time will be:
	\[ \TO\underbrace{\left(m+\frac{mn}{\tau} + \frac{m\tau n}{\Delta}\right)}_{\text{query time}} 
	+ \underbrace{\TO\left( \frac{n^2}{\tau} + n\Delta + \frac{n^2\tau}{\Delta}  \right)}_{\text{preprocessing time}} + O\left(k \right)\]
	We set the parameters $\tau \leftarrow n^{1/3}$ 
	and $\Delta \leftarrow n^{2/3}$ in the above expression 
	to obtain a running time 
	of $\TO(mn^{2/3}+ n^{5/3}) +O(k)$.

	
	\paragraph{Organization of the paper.} 
	The rest of the paper is organized as follows.
	In Section~\ref{sec:first-tec}, we will discuss
	our first technique for handling queries in $\Q$
	which have small $\LIS$. In Section~\ref{sec:second-tec}, we will discuss
	our second technique for handling queries in $\Q$
	which have large $\LIS$. In Section~\ref{sec:third-tec}, we will discuss
	our final technique for handling colors in $\mP$ 
	which are light (cardinality is small).
	In Section~\ref{sec:final-results}, we will put 
	all the techniques together in different ways 
	to  derive all the upper bounds. Next, in Section~\ref{sec:lower} we present our conditional lower bound. Finally, in Section~\ref{sec:open}, we present some open problems for future research.
	
	\section{First technique: Handling small $\LIS$}\label{sec:first-tec}
	In this section, we describe a technique 
	which efficiently handles queries with 
	small $\LIS$. 
	Consider a parameter $\tau \in [1,n]$ whose 
	value will be set later.

	
	\subsection{2D Range $\LIS$ problem}
	We will illustrate the technique by first looking at the 
	$\TRLIS$ problem. The following result is obtained.
	
	\begin{theorem}\label{thm:2d-range-small-lis}
		There is an algorithm for $\TRLIS$ problem 
		with running time $O(m\tau\log^3n + n\tau\log^4n + k)$,
		where $k$ is the cumulative length of all the 
		output subsequences. For all queries $\query \in \Q$ 
		with $\LIS(\mP \cap q) \leq \tau$, the algorithm 
		returns the correct solution. For queries 
		$\query \in \Q$ with $\LIS(\mP \cap q) > \tau$,
		correctness is not guaranteed.
	\end{theorem}

	The preprocessing phase of
	the algorithm consists of the following steps.
	
	\paragraph{Lowest peaks and highest bases.}
	Consider the special case where there is a 
	vertical line with $x$-coordinate  $x^*$ such that each query range in $\Q$ intersects $x^*$.
	Let $\LP \subseteq \mP$ (resp., $\RP \subseteq \mP$)
	be the set of points with $x$-coordinate less than 
	or equal to (resp., greater than) to $x^*$.
	For an increasing subsequence $S \in \mP$, we define {\em peak} 
	(resp., {\em base}) to be the last 
	(resp., first) element in $S$.
	We will now define two arrays $A$ and $B$ as follows:
	\begin{itemize}
		\item For all $p_i \in \LP$
		and for all $1\leq \alpha \leq \tau$,
		among all increasing subsequences of length
		$\alpha$ in the range $[i, x^*]$ which 
		have $p_i$ as the {\em first} element, 
		let $S$ be the subsequence with the 
		{\em lowest} peak. Then $A(p_i,\alpha)$
		stores the value of the last element in $S$.
		\item For all $p_i \in \RP$ and 
		for all $1\leq \beta \leq \tau$, 
		among all increasing subsequences of length 
		$\beta$ in the range $(x^*, i]$ which have
		$p_i$ as the {\em last} element, let $S$ be 
		the subsequence with the {\em highest} base.
		Then $B(p_i,\beta)$ stores the value of the 
		first element in $S$.
	\end{itemize}

	\paragraph{Efficient computation of $B$.}
	The next step is to compute each entry in 
	$B$ in polylogarithmic time (analogous discussion 
	holds for $A$). 
	The entries will be computed in increasing values
	of $\beta$. 
	Assume that the entries 
	corresponding to sequences of length at most $\beta {-}1$ have been computed. Then the entries 
	$B(p_i,\beta)$'s can be reduced to entries of
	$B(\cdot,\beta-1)$'s as follows:
	\begin{align*}
		B(p_i,\beta)=
		\begin{cases}
			a_{i}, \quad \text{ if $\beta=1$;} \\
			-\infty, \ \ \ \text{if $\beta > 1$ and 
				$a_j > a_i$ for all $\RP \cap (x^*, i)$;} \\
			\max\{B(p_j,\beta{-}1) \mid \ x^*< j < i \text{ and } 
			a_j < a_i\},
			\quad   \text{otherwise}.
		\end{cases}
	\end{align*}
	We will reduce the problem of computing 
	$B(p_i,\beta)$ to the problem of constructing 
	a data structure which can efficiently answer
	{\em 2D range-max queries}. 
	In a 2D range-max problem, we are given $n$ weighted points 
	$\RP$ in 2D. Each point $p$ is associated with a weight.
	For each point $p_i=(i,a_i)\in \RP$, 
	given a query region of the form 
	$\query'=(-\infty, i) \times (-\infty, a_i)$, 
	the goal is to report the point in $ \RP \cap \query'$
	with the maximum weight.
	The 2D range-max problem can be solved in 
	$O(n\log n)$ time (via reduction to the so-called
	{\em 2D orthogonal 
		point location} problem~\cite{st86}).

	Now with each point $p_j\in \RP$, we 
	associate weight $B(p_j,\beta-1)$
	and build a 2D range-max data structure.
	For each point $p_{i} =(i,a_{i}) \in \RP$, the point
	in $\RP \cap 
	((-\infty, i) \times (-\infty, a_{i}))$ with 
	the maximum weight is reported. If point $p_j$ is reported, 
	then set $B(p_i,\beta)=B(p_j,\beta-1)$.
	The correctness follows from the third case in the above 
	dynamic programming.
	Since $\beta \leq \tau$, the time taken 
	to construct $B$ is
	$O(\tau)\times O(n\log n)=O(n\tau\log n)$.

	\paragraph{Data structures to compute $L$ and $R$.}
	During the query algorithm, we will need to compute 
	some of the entries in two-dimensional arrays
	$L$ and $R$ which are defined as follows:
	\begin{description}
		\item[Array $L$:] For any $\query \in \Q$
		and for any 
		$1\leq \alpha \leq \tau$, among 
		all increasing subsequences of length 
		$\alpha$ in $\LP \cap \query$, let $S$ be the subsequence with the {\em lowest} peak.
		Then $L(\query,\alpha)$ is the value of the last element in $S$.
		\item[Array $R$:] For any $\query \in \Q$ and for any
		$1\leq \beta \leq \tau$, among all 
		increasing sequences of length $\beta$ in $\RP \cap \query$,
		let $S$  be the subsequence  with the {\em highest} base. 
		Then $R(\query,\beta)$ is the value of the 
		first element in $S$.
	\end{description}
	
	Each entry in $R$ is connected to entries 
	in $B$ as follows:  
	\[R(q,\beta) \leftarrow \max_{p_i}
	\{ B(p_i,\beta) \mid p_i\in \RP \cap \query \text{ and } B(p_i,\beta) > y_1 \}.\]
	Analogously, each entry in $L$ is 
	connected to entries in $A$ as follows:
	\[L(q,\alpha) \leftarrow \max_{p_i}
	\{ A(p_i,\alpha) \mid p_i\in \LP \cap \query \text{ and } A(p_i,\alpha) < y_2 \}.\]
	
	Fix a value of $\beta$. We will now build a 
	data structure to efficiently compute $R(\query, \beta)$'s.
	For each point $p_i=(i,a_i) \in \RP$, map it
	to a point $p_i'=(i, a_i, B(p_i, \beta))$ 
	with weight $w(p_i') \leftarrow B(p_i,\beta)$.
	Let $\RP'$ be the collection of mapped 3D points.
	Build a 3D vanilla range tree~\cite{bcko08} on $\RP'$.
	Given a query $\query \in \Q$, it is transformed 
	to a 3D cuboid $\query'=\query \times (y_1,+\infty)$ and
	the point with the maximum weight in 
	$\RP' \cap \query'$ is reported efficiently.
	The time taken to construct the range tree is 
	$O(n\log^3n)$. We will repeat this procedure
	for all values of $\beta \in [1,\tau]$.
	Therefore, the total construction time will be
	$O(n\tau \log^3n)$. Analogously, 
	construct a data structure to efficiently 
	compute $L(\query, \alpha)$'s.

	\paragraph{Recursion.} Let $\D$ be the 
	data structure built above to handle
	queries in $\Q$ which intersect the vertical 
	line with $x$-coordinate $x^*$. 
	We will handle the general case via recursion.
	Let $x^*=n/2$, and let $\mP_{\leq} \subseteq \mP$ 
	(resp., $\mP_{>} \subseteq \mP$) be the set of 
	points with $x$-coordinate less than or equal to 
	(resp., greater than) $x^*$. Recursively
	build data structure $\D_{\leq}$ (resp., $\D_{>}$)
	based on pointset $\mP_{\leq}$ (resp., $\mP_{>}$).
	The base case of $|\mP|=1$ can be 
	handled trivially. Let $T(n)$ be the total 
	preprocessing time. Then, 
	\[T(n) \leq 2\cdot T(n/2) + O(n\tau\log^3n),\]
	which solves to $T(n)=O(n\tau\log^4n)$.

	\begin{lemma}
		The preprocessing time of the algorithm is 
		$O(n\tau\log^4n)$.
	\end{lemma}
	
	\paragraph{Query algorithm.}
	The query algorithm consists of the following steps.
	Let $\Q_{x^*} \subseteq \Q$ be the set of queries
	which intersect the vertical line with $x$-coordinate
	$x^*=n/2$. We will first handle $\Q_{x^*}$. For each $\query \in \Q_{x^*}$,  
	compute $R(\query,\beta)$, for all $1\leq \tau \leq \beta$, 
	by querying the corresponding 3D range trees.
	Analogously, compute $L(\query, \alpha)$ 
	for all $1\leq \tau \leq \alpha$. 
	
	A pair $\left(L(\query, \alpha), R(\query, \beta)\right)$
	is defined
	to be {\em compatible} if $L(\query,\alpha) < R(\query, \beta)$, i.e., 
	it is possible to ``stitch'' the 
	sequences corresponding to $L(\query,\alpha)$ and 
	$R(\query, \beta)$ to obtain 
	an increasing sequence of length $\alpha +\beta$.
	For a 2D query range $\query$, 
	our goal is to find the compatible pair 
	$\left(L(\query, \alpha), R(\query, \beta)\right)$ which 
	maximizes the value of $\alpha + \beta$, i.e., 
	\begin{equation}\label{eq:query-small-lis}
		\LIS(\mP \cap \query)=\max_{\alpha,\beta}\{\alpha + \beta   \mid 
		L(\query, \alpha)< R(\query, \beta)\}.
	\end{equation}
	
	To efficiently compute $\LIS(\mP \cap \query)$, 
	we will use the {\em monotonicity} property of 
	$R$: for a fixed query $\query$, the value of 
	$R(\query, \beta)$
	decreases as the value of $\beta$ increases, i.e., 
	for any $1\leq \beta \leq \beta' \leq \tau$, we have 
	$R(\query, \beta) \geq R(\query, \beta')$. 
	For each $1\leq i \leq \tau$, 
	the largest $\beta_i$ for which 
	$(L(\query,i), R(\query, \beta_i))$
	is compatible can be found in $O(\log \tau)$ time by doing 
	a binary search on 
	the monotone sequence 
	$(R(\query,1), R(\query, 2),\ldots, R(\query, \tau))$.
	Finally,
	report $\max_{i}(i+\beta_i)$ as the length of the 
	$\LIS$ for 
	$\mP \cap \query$. By appropriate bookkeeping, the corresponding $\LIS$
	can be reported in time proportional to its length.
	
	Let $\Q_{<} \subseteq \Q$ (resp., $\Q_{>} \subseteq \Q$)
	be the set of queries which lie completely to 
	the left (resp., right) of the vertical line with 
	$x$-coordinate $x^*$.
	Finally, recurse on pointset $\mP_{\leq}$ and 
	query set $\Q_{<}$, and 
	recurse on pointset $\mP_{>}$ and query set $\Q_{>}$.

	\begin{lemma}
		The time taken answer the query ranges in 
		$O(m\tau\log^3n + k)$.
	\end{lemma}
	\begin{proof}
		Fix a query range $q\in \Q$. Querying the
		3D range trees takes $O(\log^3n)$ time. 
		As such, the time taken to 
		compute $R(\query, \beta)$'s and $L(\query, \alpha)$'s 
		will be $O(\tau\log^3n)$. Next, computing the 
		compatible pair which achieves the quantity
		$\max_{i}(i+\beta_i)$ takes $O(\tau\log \tau)$ time.
		Therefore, the time spent per query is $O(\tau\log^3n + \tau\log\tau)=O(\tau\log^3n)$. 
		
		Assigning queries in $\Q$ 
		to appropriate nodes in the recursion tree takes only
		$O(m\log n)$ time and is not the dominating term.
	\end{proof}

	
	\subsection{Colored 2D Range $\LIS$ problem}
	The solution for $\CRLIS$ is obtained by 
	solving the $\TRLIS$ problem independently 
	for each color. Specifically,
	for each color $c\in \C$, solve the (uncolored)
	$\TRLIS$ problem for $\mP_c$ using 
	Theorem~\ref{thm:2d-range-small-lis}. The length 
	version of the $\TRLIS$ problem will output 
	for each $\query \in \Q$ and each $c\in \C$, the value of 
	$\LIS(\mP_c \cap \query)$. Then for each 
	$\query \in \Q$, report the color $c_{\query}$
	which maximizes the value of $\LIS(\mP_{c_{\query}} \cap \query)$ and the corresponding $\LIS$. 
	Ignoring the $k$ term, the running time of the algorithm 
	is $|\C|$ times the running time of the $\TRLIS$ problem.
	
	\begin{theorem}\label{thm:col-2d-small-lis}
		There is an algorithm for $\CRLIS$ problem 
		with running time $O(m\tau|\C|\log^3n + n\tau|\C|\log^4n + k)$,
		where $k$ is the cumulative length of all the 
		output subsequences. For all queries $\query \in \Q$ 
		with $\LIS(\mP_{c_{\query}} \cap q) \leq \tau$, 
		the algorithm 
		returns the correct solution. For queries 
		$\query \in \Q$ with $\LIS(\mP_{c_{\query}} \cap q) > \tau$,
		correctness is not guaranteed.
	\end{theorem}

	\section{Second technique: Handling large $\LIS$}
	\label{sec:second-tec}
	In this section, we describe a technique 
	which efficiently handles queries with large $\LIS$.
	
	\subsection{2D Range $\LIS$ problem}
	We will first consider the
	$\TRLIS$ problem. Consider a parameter $\tau \in [1,n]$ 
	whose value will be set later. The following result is obtained.

	\begin{theorem}\label{thm:large-LIS-2d}  
		There is an algorithm for the $\TRLIS$ problem
		with running time $O\left(\frac{mn}{\tau}\log^5n + 
		\frac{n^2}{\tau}\log^4n + k \right)$,
		where $k$ is the cumulative length of all the 
		output subsequences. The bound on the running 
		time holds with high probability. 
		For all queries $\query \in \Q$ with 
		$\LIS(\mP \cap q) \geq \tau$, with high probability, the algorithm returns the correct solution.
	\end{theorem}
	
	\paragraph{Idea-1: Stitching elements.}
	Suppose an oracle reveals that point $p_i\in \mP$ lies in the 
	$\LIS$ of  range $\query=[x_1,x_2] \times [y_1,y_2]$. 
	Then we claim that the problem becomes decomposable.
	We will define some notation before proceeding. 
	For a point $p_i=(i,a_i) \in \mP$ corresponding to 
	an element $a_i \in \mS$, define the {\em north-east} region of $p_i$ as $NE(p_i) = \{ (x,y) \mid x \geq i \text{ and } y \geq a_i \}$.
	Analogously, define the {\em north-west}, the 
	{\em south-west} and the {\em south-east} 
	region of $p_i$ which are denoted by 
	$NW(p_i), SW(p_i) \text{ and } SE(p_i)$, 
	respectively.
	Then it is easy to observe that,
	\begin{equation}\label{eq:decomposable-main}
		\LIS(\mP \cap \query)=\LIS(\mP \cap \query \cap 
		NE(p_i)) 
		+ \LIS(\mP \cap \query \cap SW(p_i)) -1.
	\end{equation}
	See Figure~\ref{fig:stitch-2D}
	for an example. In other words, knowing that $p_i$ belongs to the $\LIS$ decomposes the 
	original $\LIS$ problem into two  $\LIS$ sub-problems which can be 
	computed independently. 
	In such a case, we refer to point $p_i$ as a {\em stitching element},
	which is formally defined below.
	\begin{definition}
		{\bf(Stitching element)} For a range $\query$, 
		fix any $\LIS$ of $\mP \cap \query$ and  call it $S$. 
		Then each element in $S$ is defined to be a stitching element w.r.t. $\query$.
	\end{definition}
	
	The goal of our algorithm is to:
	\begin{center}
		{\em Construct a small-sized set $\R \subseteq \mP$ 
			such that, for any $\query\in \Q$,  at least one 
			stitching element w.r.t. $\query$ is contained in 
			$\R \cap \query$.}
	\end{center}
	

	\begin{figure*}
		\centering
		\includegraphics{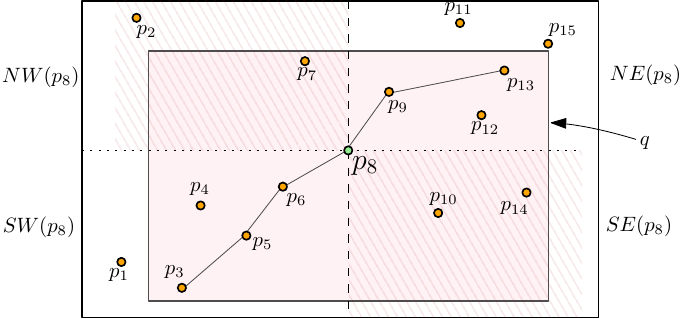}
		\caption{A collection of fifteen elements. 
			Given a query range $q$ (shown in pink), suppose an oracle reveals
			that $p_8$ lies in the $\LIS$ of $\mP \cap \query$.
			Then observe that we can discard the points 
			in $SE(p_8)$ and $NW(p_8)$ which are shown 
			as shaded regions.  The $\LIS$ of 
			$\mP \cap \query \cap NE(p_8)$ will be 
			$(p_8,p_9, p_{13})$ and and the $\LIS$
			of $\mP \cap \query \cap SW(p_8)$ will be 
			$(p_3,p_5,p_6,p_8)$. Therefore, $\LIS$ of 
			$\mP \cap \query$ will be 
			$(p_3, p_5, p_6, p_8, p_9, p_{13})$.}
		\label{fig:stitch-2D}
	\end{figure*}

	\paragraph{Idea-2: Random sampling.} By using the fact 
	that $\LIS(\mP \cap q) \geq \tau$, it is possible 
	to construct a set $\R$ of size 
	$O\left(\frac{n\log n}{\tau}\right)$ via random sampling.
	For each $1\leq i\leq n$, sample  $p_i \in \mP$ independently 
	with probability $\frac{c\log n}{\tau}$, where $c$ is a sufficiently 
	large constant. Let $\R \subseteq \mP$ be the set of sampled points. 
	Then we can establish the following 
	connection between $\R$ and the stitching elements.
	
	\lemlogn*
	\begin{proof}
		Fix a range $\query\in \Q$. 
		For each point $e\in \sq$, 
		let $X_e$  be an indicator random 
		variable which is one if $e \in \R$, otherwise it is zero. Next, define
		another random variable $X=|\R \cap \sq|$. Then,
		\begin{align*}
			{\bf E}[X]=\sum_{e\in \sq} {\bf E}[X_e] = |\sq|\cdot p 
			\geq \tau \cdot \frac{c\log n}{\tau} 
			=c\log n,
		\end{align*}
		where we used the fact $|\sq| \geq \tau$. Next, 
		for a sufficiently large constant $c'$, 
		we observe that
		\begin{align*}
			{\bf Pr}\left(X < \frac{c\log n}{2}\right) \leq 
			{\bf Pr}\left(X < \frac{{\bf E}[X]}{2}\right) \leq 
			e^{-\frac{{\bf E}[X]}{8}} \leq \frac{1}{n^{c'}},
		\end{align*}
		where the second inequality follows by setting 
		$\delta=1/2$ in the following version of Chernoff 
		bound: ${\bf Pr}(X < (1-\delta){\bf E}[X]) \leq e^{-\delta^2 {\bf E}[X]/2}$, 
		for $0 < \delta < 1$. 
		Via a straightforward application of the union bound, 
		with high probability none of the 
		ranges $\query\in \Q_L$ can have less than $\frac{c\log n}{2}$ 
		elements in $\R \cap \sq$ (when $c\gg 16$). 
	\end{proof}
	
	\paragraph{Computing restricted-$\LIS$'s efficiently.}
	For each query $\query \in \Q$, if $p_i \in \R$ lies inside
	$\query$, then we want to efficiently compute 
	the quantities $\LIS(\mP \cap \query \cap 
	NE(p_i))$ and $\LIS(\mP \cap \query \cap SW(p_i))$ 
	(as shown in equation~\ref{eq:decomposable-main}). 
	To enable that, we will now compute {\em restricted-$\LIS$'s} between some specifically 
	chosen pairs of points.
	
	Consider a point $p_i=(i,a_i) \in \R$. 
	Scan the points $\mP \cap NE(p_i)$ in increasing
	order of their $x$-coordinate value. We will assign a
	weight $w(p_j)$ for each point encountered. 
	As a base case, we will assign $w(p_i) \leftarrow 1$.
	At a general step, if we encounter point $p_j$, then we 
	\begin{equation}\label{eq:restricted-lis}
		w(p_j) \leftarrow 1 + \max_{p_{\alpha}=(\alpha,a_{\alpha}) \in (\mP \cap NE(p_i))}
		\{ w(p_{\alpha}) \mid \alpha < j \text{ and } a_{\alpha} < a_j \}.
	\end{equation}
	See Figure~\ref{fig:restricted-lis} for an 
	example. Repeat this process for each point in 
	$\R$. Now we claim the following.

	\begin{lemma}\label{lem:ne-max}
		For a given query $\query \in \Q$, let $p_i \in \R
		\cap \query$. Then $\LIS(\mP \cap q \cap NE(p_i))
		\leftarrow \max_{p_j \in (\mP \cap q \cap NE(p_i))}\{w(p_j)\}$.
	\end{lemma}
	\begin{proof}
		The correctness follows from the fact that 
		$w(p_j)$ captures the length of the $\LIS$
		which has $p_i$ (resp., $p_j$) as the first (resp., last)
		point.
	\end{proof}
	\begin{figure}
		\centering
		\includegraphics{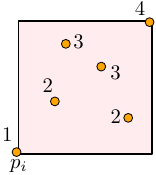}
		\caption{$p_i$ is a point in $\R$ 
			and $NE(p_i)$ has six points. The weight computed for each point by Equation~\ref{eq:restricted-lis} is shown 
			besides each point.}
		\label{fig:restricted-lis}
	\end{figure}
	
	\begin{lemma}\label{lem:pi-entries}
		The computation of restricted-$\LIS$'s, i.e., 
		computation of
		$w(p_j)$, for all $p_j \in \mP \cap NE(p_i)$
		can be done in $O(n\log n)$ time.
	\end{lemma}
	\begin{proof}
		Perform a linear scan of $\mP$ to identify
		the points lying in $\mP \cap NE(p_i)$.
		Then any standard $O(n\log n)$ time algorithm
		to compute $\LIS$ can be run
		on $\mP \cap NE(p_i)$ to compute $w(p_j)$'s.
	\end{proof}

	\begin{lemma}\label{lem:all-pi-entries}
		There is an algorithm,  which with high probability,
		computes restricted-$\LIS$'s for all points in 
		$\R$ in
		$O\left(\frac{n^2}{\tau}\log^2n\right)$ time. 
	\end{lemma}
	\begin{proof}
		For all $1\leq i\leq n$, let $X_i$ be the random variable 
		which is one if $p_i \in \R$, otherwise it is zero.
		Let $T$ be the random variable which is the time taken 
		to compute restricted-$\LIS$'s, for all  points in 
		$\R$. Then, by Lemma~\ref{lem:pi-entries}, we have
		\[T=O(n\log n)\cdot \sum_{i=1}^{n} X_i +O(n).\] 
		Now consider the random variable $X$ which is equal to 
		$\sum_{i=1}^{n} X_i$. Then ${\bf E}[X]={\bf E}[\sum_{i=1}^{n} X_i]=\frac{cn}{\tau}\log n$.
		Let $c'$ be a sufficiently large constant.
		By using an appropriate version of Chernoff bound, we observe that
		\[ {\bf Pr}\left(X> \frac{2cn}{\tau}\log n\right) \leq {\bf Pr}(X> 2{\bf E}[X]) < e^{-\Omega({\bf E}[X])} < \frac{1}{n^{c'}}, \]
		where the last inequality used the trivial fact that 
		${\bf E}[X]\geq c\log n$.
		Therefore, with high probability $T$ is bounded by $O\left(\frac{n^2}{\tau}\log^2n\right)$.
	\end{proof}
	
	We will perform an analogous procedure for computing 
	$\LIS(\mP \cap \query \cap SW(p_i))$, 
	for each $p_i \in \R$.
	
	\paragraph{Representing restricted-$\LIS$s 
		as rectangles.}
	Consider a point $p_i \in \R$. 
	For each point $p_j \in \mP \cap NE(p_i)$, 
	let $B_{ij}$ be an axis-parallel rectangle 
	with $p_i$ as the lower-left corner and
	$p_j$ as the top-right corner, with an 
	associated weight of $w(p_j)$. Let 
	$\B_i$ the collection of these rectangles.

	In a {\em rectangle range-max query}, the input is a
	set $\B_i$ of weighted axis-aligned rectangles in 2D. 
	Given a query rectangle $\query$, among all the rectangles
	in $\B_i$ which lie completely inside $\query$, 
	the goal is to report 
	the rectangle  with the largest weight. 
	We will use (vanilla) 4D range trees
	as our data structure~\cite{bcko08}.
	The data structure can be constructed in 
	$O(|\B_i|\log^3|\B_i|)$ 
	time and the rectangle range-max query can be answered in 
	$O(\log^4 |\B_i|)=O(\log^4 n)$ time. 
	The bounds hold with high probability. 

	\begin{lemma}
		The preprocessing time of the algorithm is 
		$O\left(\frac{n^2}{\tau}\log^4n\right)$.
		The bound holds with high probability.
	\end{lemma}
	\begin{proof}
		By Lemma~\ref{lem:all-pi-entries}, construction
		of all the restricted-$\LIS$'s takes
		$O\left(\frac{n^2}{\tau}\log^2n\right)$ time.
		The preprocessing time is dominated by the 
		construction time of sets $\B_i$, for all 
		$p_i \in \R$, which is $\sum_{p_i \in \R}O(|\B_i|\log^3|\B_i|)
		=O\left(\frac{n}{\tau}\log n\right)\cdot O(n\log^3n)= 
		O\left(\frac{n^2}{\tau}\log^4n\right)$.
	\end{proof}
	
	\paragraph{Remark.} The bound of $O\left(\frac{n^2}{\tau}\log^4n\right)$ in the above lemma can be improved
	by one or two log factors by using more sophisticated data structures. However, as pointed out in the Introduction, 
	optimizing the log factors is not the focus of this paper.

	The following lemma establishes the connection 
	between Lemma~\ref{lem:ne-max} and set $\B_i$.
	\begin{lemma}\label{lem:ne-max-box}
		For a given query $\query \in \Q$, let $p_i \in \R
		\cap \query$. Then $\LIS(\mP \cap q \cap NE(p_i))$
		is equal to the weight of the 
		largest weighted rectangle in $\B_i$
		which lies completely inside $\query \cap NE(p_i)$.
	\end{lemma}
	\begin{proof}
		The proof follows from the fact that a 
		point $p_j$ belongs to $\mP \cap q \cap NE(p_i)$ 
		if and only if rectangle $B_{ij}$ lies completely inside 
		$\query$. 
	\end{proof}

	\paragraph{Query algorithm.}
	Consider any range $\query \in \Q$. Scan $\R$ to 
	identify the points 
	which lie inside $\query$. If $|\R \cap \query|=\emptyset$, 
	then we do not proceed further for $\query$.
	Otherwise, for each element $p_i \in \R \cap \query$,
	we want to compute 
	\[\LIS(\mP \cap \query \cap NE(p_i)) 
	+ \LIS(\mP \cap \query \cap SW(p_i)) -1.\]
	To compute $\LIS(\mP \cap \query \cap NE(p_i))$, 
	we use Lemma~\ref{lem:ne-max-box} and pose a
	rectangle range-max query on $\B_i$ with 
	query $\query \cap NE(p_i)$. Analogously, 
	we compute $\LIS(\mP \cap \query \cap SW(p_i))$.
	Finally, report the largest value computed.
	By Lemma~\ref{lem:logn-stitching}, we claim that
	the answer is correct with high probability if 
	$\LIS(\mP \cap \query)\geq \tau$. Repeat this procedure for 
	each $\query \in \Q$. The query time 
	will be $O(|\Q||\R|\log^4n)$ which will be $O\left(\frac{mn}{\tau}\log^5n \right)$ with high probability.
	By appropriate bookkeeping, reporting the $\LIS$ of $\mP\cap \query$, for all $\query \in \Q$
	can be done in $O(k)$ time.
	
	\begin{lemma}
		The query time of the algorithm is $O\left(\frac{mn}{\tau}\log^5n +k \right)$. With high probability, 
		the correctness and the bound on the running 
		time holds.
	\end{lemma}
	
	This finishes the proof of Theorem~\ref{thm:large-LIS-2d}.

	\subsection{Colored 2D range $\LIS$ problem}
	In the $\CRLIS$ problem, 
	for each $\query \in \Q$, let  
	$c_{\query} \in \C$ be the color for which $\LIS(\mP_{c_{\query}} \cap \query)$ is maximized. The algorithm 
	for $\CRLIS$ requires two modifications 
	to the algorithm for $\TRLIS$. The first
	modification is the precise connection 
	between $\R$ and the stitching elements.
	
	\begin{lemma}\label{lem:logn-stitching-col}
		For all ranges $\query$ such that 
		$\LIS(\mP_{c_{\query}} \cap q) \geq \tau$, 
		if $\sq$ is one of the $\LIS$ of $\mP_{c_{\query}} \cap \query$, then 
		with high probability
		$|\R \cap \sq|=\Omega(\log n)$.
	\end{lemma}
	\begin{proof}
		The proof follows directly from the proof of Lemma~\ref{lem:logn-stitching} by replacing $\mP$ with $\mP_{c_{\query}}$.
	\end{proof}
	Let $p_i\in \R$ have color $c$.
	Next, we modify Equation~\ref{eq:restricted-lis}
	by replacing $\mP$ with $\mP_c$. Specifically, we do 
	the following:
	\begin{equation}\label{eq:restricted-lis-col}
		w(p_j) \leftarrow 1 + \max_{p_{\alpha}=(\alpha,a_{\alpha}) \in (\mP_c \cap NE(p_i))}
		\{ w(p_{\alpha}) \mid \alpha < j \text{ and } a_{\alpha} < a_j \}.
	\end{equation}
	
	The remaining preprocessing steps and the 
	query algorithm of $\TRLIS$ can be trivially adapted
	for the $\CRLIS$ problem. 
	The final result is summarized below.
	\begin{theorem}\label{thm:large-LIS-col-2d}  
		There is an algorithm for the $\CRLIS$ problem
		with running time $O\left(\frac{mn}{\tau}\log^5n + 
		\frac{n^2}{\tau}\log^4n + k \right)$,
		where $k$ is the cumulative length of all the 
		output subsequences. The bound on the running 
		time holds with high probability. 
		For all queries $\query \in \Q$ with  
		$\LIS(\mP_{c_{\query}} \cap q) \geq \tau$, 
		with high probability, the algorithm returns the correct solution.
	\end{theorem}

	\section{Third technique: Handling small cardinality
		colors}\label{sec:third-tec}
	In this section we will discuss a technique 
	to handle colored $\RLIS$ problems. The algorithm
	is efficient when each color class has a small cardinality.
	We will prove the following two results.
	
	\begin{theorem}\label{thm:third-tec-2D}
		Fix a parameter $\Delta$.
		Then there is a deterministic algorithm to answer
		$\CRLIS$ problem in $O(m\log^4n + n\Delta\log^4n + k)$ time,
		where for each color $c \in \C$, we have $|\mP_c|\leq \Delta$, and $k$ is the cumulative length of the $m$
		output subsequences.
	\end{theorem}
	
	\begin{theorem}\label{thm:third-tec-1D}
		Fix a parameter $\Delta$.
		Then there is a deterministic algorithm to answer
		$\CORLIS$ problem in $O(m\log n + n\Delta\log^2 n + k)$ time,
		where for each color $c \in \C$, we have $|\mP_c|\leq \Delta$, and $k$ is the cumulative length of the $m$
		output subsequences.
	\end{theorem}
	
	We will first consider the $\CRLIS$ problem.
	Fix a color $c$. 
	For a query $\query$, consider the $\LIS$ of 
	points of $\mP_c \cap q$. 
	Let  $p_i \in \mP_c$ (resp., $p_j \in \mP_c$) be the 
	first (resp., last) point in the $\LIS$. Call this a pair $(i,j)$.
	As such, the number of distinct pairs
	of color $c$ will be $O(|\mP_c|^2)$. Adding up 
	over all the colors, the total 
	number of distinct pairs will be:
	\[\sum_{c}O(|\mP_c|^2) \leq O\left(\Delta\sum_{c}|\mP_c|\right)=O(n\Delta).\]
	For each color $c$ and for each of its pair $(i,j)$, 
	we will precompute and store the length of the $\LIS$
	starting with $p_i$ and ending with $p_j$.
	See Figure~\ref{fig:all-pairs} for an example.
	
	\paragraph{Constructing the set $\R$.} 
	Consider a color $c$ and a point $p_i=(i,a_i) \in \mP_c$.
	Define $NE(p_i)$ to be set of points in 2D which lie
	in the {\em north-east} region of $p_i$, i.e.,
	$NE(p_i)=\{(p_x,p_y) \mid p_x > i \text{ and } p_y > a_i\}$.
	Consider the points $\mP_c \cap NE(p_i)$ in increasing
	order of their $x$-coordinate value. We will assign a
	weight $w(p_j)$ for each point encountered. 
	As a base case, we will assign $w(p_i) \leftarrow 1$.
	At a general step, if we encounter point $p_j$, then we 
	\begin{equation}\label{eq:all-pairs}
		w(p_j) \leftarrow 1 + \max_{p_{\alpha}=(\alpha,a_{\alpha}) \in (\mP_c \cap NE(p_i))}
		\{ w(p_{\alpha}) \mid \alpha < j \text{ and } a_{\alpha} < a_j \}.
	\end{equation}
	See Figure~\ref{fig:all-pairs} where
	$p_j$ is assigned a weight of four.
	Construct an axis-aligned rectangle 
	$R_c(i,j)$ with $p_i$ (resp., $p_j$) as the
	bottom-left (resp., top-right) corner. 
	A weight of $w(p_j)$ is associated with 
	$R_c(i,j)$. The intuition is that  $w(p_j)$ is equal to
	the $\LIS$ of the points of color $c$ lying inside 
	$R_c(i,j)$, i.e., $\LIS(\mP_c \cap R_c(i,j))$.
	
	\begin{figure}
		\centering
		\includegraphics{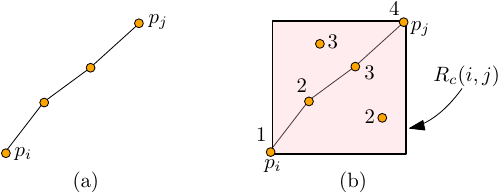}
		\caption{(a) An increasing subsequence 
			of length four with all four points
			being of color $c$. (b) The weight computed by Equation~\ref{eq:all-pairs} is shown 
			besides each point. The $\LIS$ of color $c$ inside
			$R_c(i,j)$ is four (the other two points of color 
			$c$ do not participate in the $\LIS$). Therefore,
			the weight associated with $R_c(i,j)$ is four. }
		\label{fig:all-pairs}
	\end{figure}
	The computation of 
	$w(p_j)$, for all $p_j \in \mP_c \cap NE(p_i)$
	can be done in $O(|\mP_c|\log^2 |\mP_c|)$ time.
	Construct a vanilla 2D range-tree~\cite{bcko08} based 
	on the points $\mP_c \cap NE(p_i)$.
	Set the weight of $p_i$ equal to one and 
	the weight of remaining points to $-\infty$.
	Whenever a point $p_j$ is encountered, then 
	query the range tree with query rectangle 
	$q'=(-\infty, j) \times
	(-\infty, a_j)$ and report the point in 
	$\mP_c \cap NE(p_i)\cap q'$
	with the maximum weight. The query time 
	is $O(\log^2|\mP_c|)$. If $p_{\alpha}$ is the 
	point reported, update the range 
	tree with the new weight of 
	$w(p_j) \leftarrow 1+w(p_{\alpha})$.
	The update takes $O(\log^2|\mP_c|)$ amortized time.

	Repeat this procedure for each point in $\mP_c$.
	This takes $O(|\mP_c|^2\log^2 |\mP_c|)$ time.
	Define $\R_c \leftarrow \bigcup_{(i,j)}R_c(i,j)$
	and $\R\leftarrow \bigcup_{c} \R_c$ which is a collection of rectangles corresponding to the distinct 
	pairs $(i,j)$. The total time taken to 
	construct $\R$ will be $\sum_{c}O(|\mP_c|^2\log^2 |\mP_c|)
	=O(n\Delta\log^2n)$.
	
	\paragraph{Rectangle range-max structure.}
	In a {\em rectangle range-max query}, the input is a
	set $\R$ of weighted axis-aligned rectangles in 2D. 
	Given a query rectangle $\query$, among all the rectangles
	in $\R$ which lie completely inside $\query$, 
	the goal is to report 
	the rectangle  with the largest weight. 
	Note that in the rectangle range-max query, the 
	color of the rectangles does {\em not} matter.
	We are dealing with uncolored rectangles 
	and as a result, we can use vanilla 4D range trees
	as our data structure~\cite{bcko08}.
	The data structure can be constructed in 
	$O(|\R|\log^4|\R|)=O(n\Delta\log^4n)$ 
	time and the rectangle range-max query can be answered in 
	$O(\log^4 |\R|)=O(\log^4 n)$ time.
	
	Coming back to our $\CRLIS$ problem, for a query range $\query \in \Q$, 
	we first query the rectangle range-max data structure.
	Let $c_{\query}$ be the color
	corresponding to the reported rectangle. We claim that
	$c_{\query}$ is the color 
	for which  $\LIS(\mP_{c_{\query}}\cap \query)$ 
	is maximized. Overall, answering $m$ range queries
	in $\Q$ requires only $O(m\log^4n)$ time.
	By appropriate bookkeeping, reporting the $\LIS$ of $\mP_{c_{\query}}\cap \query$, for all $\query \in \Q$
	can be done in $O(k)$ time. This finishes the
	proof of Theorem~\ref{thm:third-tec-2D}.

	\paragraph{Shaving log factors in 1D.}
	To answer $\CORLIS$ the rectangle range-max
	data structure can be replaced by 
	{\em interval range-max} data structure.
	For each rectangle $R\in \R$, let $I$ be the 
	projection of $R$ onto the $x$-axis. The 
	weight of $I$ is equal to the weight of $R$,
	for all $R\in \R$. Let $\I \leftarrow \bigcup_{R\in \R} I$
	be the collection of weighted intervals.
	In an interval range-max query, the input is a
	set $\I$ of weighted intervals on the real line. 
	Given a query range $\query$, among all the intervals
	in $\I$ which lie completely inside $\query$, 
	the goal is to report 
	the interval  with the largest weight. 
	
	The interval range max data structure can be 
	constructed in $O(|\I|\log |\I|)=O(n\Delta\log n)$ time 
	(via reduction to {\em 2D orthogonal point location}
	problem~\cite{st86}) and the query can be 
	answered in $O(\log |\I|)=O(\log n)$ time. 
	This finishes the proof of Theorem~\ref{thm:third-tec-1D}.

	\section{Putting all the techniques together}
	\label{sec:final-results}
	
	Finally, in this section we will put together
	all the three techniques to obtain our upper bound
	results.

	\subsection{2D Range $\LIS$ problem}
	The algorithm for $\TRLIS$ 
	applies the first technique on $\mP$ and $\Q$, 
	and then applies the second technique on 
	$\mP$ and $\Q$. For each query $\query \in \Q$,
	with high probability, the correct solution 
	is returned by one of them. 
	Using Theorem~\ref{thm:2d-range-small-lis}, the running time of the first technique 
	is $O(m\tau\log^3 n + n\tau\log^4n +k)$.
	Using Theorem~\ref{thm:large-LIS-2d}, 
	the running time of the second technique
	is $O\left(\frac{mn}{\tau}\log^5n + 
	\frac{n^2}{\tau}\log^4n + k \right)$. 
	Setting $\tau \leftarrow \sqrt{n}$, we 
	obtain a total running time of 
	$O(m\sqrt{n}\log^5n + n\sqrt{n}\log^4 + k)$.
	This proves Theorem~\ref{thm:TRLIS}.
	
	\subsection{Colored 
		1D range $\LIS$ problem}
	The algorithm consists of the following steps.
	\paragraph{Light and heavy colors.}
	Define a parameter $\Delta$ which is set to $\sqrt{n}$.
	A color $c$ is classified as {\em light} if 
	$|\mP_c|\leq \Delta$, otherwise, it is classified
	as {\em heavy}. We will design different algorithms
	to handle light colors and heavy colors.
	
	\paragraph{Querying light colors 
		simultaneously.}
	Let $\mP_{\ell} \subseteq \mP$ be the set of points
	having light color. We will use the third technique
	to answer $\CORLIS$ on $\mP_{\ell}$.
	By Theorem~\ref{thm:third-tec-1D}, the running time
	is $O(m\log n + n\sqrt{n}\log^2n + k)$.

	\paragraph{Query heavy colors independently.}
	For each heavy color, we build Tiskin's structure
	and the structure of Gawrychowski, Gorbachev, and 
	Kociumaka~\cite{GGK24}
	to answer the length version and the reporting 
	version of $\ORLIS$, respectively. Given a 
	query $\query \in \Q$, we query the length structure 
	of each heavy color and find the color $c_q$ for which 
	$\LIS(\mP_{c_{\query}}\cap \query)$ 
	is maximized. Finally, query the reporting structure
	of color $c_q$ to report the $\LIS$ of 
	$\mP_{c_{\query}}\cap \query$. 
	
	The preprocessing time (dominated by the 
	construction of reporting structures)
	is $O(n\log^3n)$.
	Since the number of 
	heavy colors is $O(\sqrt{n})$, querying the
	length structure takes  $O(\sqrt{n}\log n)$ time per 
	query, and querying the reporting structure takes
	$O(1+k_q)$ time per query, where $k_q$ is
	the length of the $\LIS$ reported. As such, the 
	time taken to answer all the $m$ queries will be 
	$O(m\sqrt{n}\log n + k)$.
	
	\paragraph{Overall algorithm.} For each 
	query $\query \in \Q$, after querying 
	the light colors and the heavy colors, 
	report the larger   among the two
	sequences reported.
	The overall running time of the algorithm is 
	$O(m\sqrt{n}\log n + n\sqrt{n}\log^2n + k)$.
	This proves Theorem~\ref{thm:CORLIS}.
	
	\subsection{Colored 2D range $\LIS$ problem}
	The algorithm for $\CRLIS$ is slightly more
	nuanced than $\CORLIS$.
	
	\paragraph{Light and heavy colors.}
	Define a parameter $\Delta$ which will be set later.
	As before, a color $c$ is classified as {\em light} if 
	$|\mP_c|\leq \Delta$, otherwise, it is classified
	as {\em heavy}. We will design different algorithms
	to handle light colors and heavy colors.

	\paragraph{Handling light colors.}
	Let $\mP_{\ell} \subseteq \mP$ 
	be the set of points which belong to a light color.
	We will use the third technique on $\mP_{\ell}$
	and queries $\Q$.
	Using Theorem~\ref{thm:third-tec-2D}, it follows that 
	the running time is $O(m\log^4n + n\Delta\log^4n + k)$.

	\paragraph{Handling heavy colors and small $\LIS$ 
		queries.}
	Let $\mP_{h} \subseteq \mP$ 
	be the set of points which belong to a heavy color.
	We will use the first technique in Section~\ref{sec:first-tec} on $\mP_h$ and queries $\Q$ (Theorem~\ref{thm:col-2d-small-lis}). The running time is
	$O\left(\frac{m\tau n}{\Delta}\log^3n + \frac{n^2\tau}{\Delta}\log^4n + k \right)$.
	
	\paragraph{Handling heavy colors and large $\LIS$ 
		queries.}
	We will use the second technique in Section~\ref{sec:second-tec} on $\mP_h$ and queries $\Q$ 
	(Theorem~\ref{thm:large-LIS-col-2d}). The running time 
	is $O\left(\frac{mn}{\tau}\log^5n + 
	\frac{n^2}{\tau}\log^4n + k \right)$.
	
	\paragraph{Overall algorithm.}
	For each 
	query $\query \in \Q$, after querying 
	each of the above three subroutines, 
	report the largest   among the three
	sequences reported.
	Combining all the three subroutines, the total running
	time will be:
	\[ \underbrace{O\left(m\log^4n+\frac{mn}{\tau}\log^5n + \frac{m\tau n}{\Delta}\log^3n\right)}_{\text{query time}} 
	+ \underbrace{O\left( \frac{n^2}{\tau}\log^4n + n\Delta 
		\log^4n + \frac{n^2\tau}{\Delta} \log^4n \right)}_{\text{preprocessing time}} + O\left(k \right)\]
	We set the parameters $\tau \leftarrow n^{1/3}$ 
	and $\Delta \leftarrow n^{2/3}$ in the above expression 
	to obtain a running time 
	of $O(mn^{2/3}\log^5n+ n^{5/3}\log^4n +k)$.
	This proves Theorem~\ref{thm:CRLIS}.
	
	\section{Conditional Lower Bound}\label{sec:lower}

	In this section we prove Theorem~\ref{thm:colorLB}. Before we do so, we first recall the Combinatorial Boolean Matrix Multiplication Hypothesis (\CBMMH) and a conditional lower bound of \cite{chan2014linear} on computing mode for range queries.

	\begin{definition}[Combinatorial Boolean Matrix Multiplication conjecture]\label{def:CBMMH}
		The Combinatorial Boolean Matrix Multiplication conjecture asserts that for every $\varepsilon>0$, no combinatorial algorithm running in time $n^{1.5-\varepsilon}$ can given as input two  $\sqrt{n}\times \sqrt{n}$ Boolean matrices, compute their product. 
	\end{definition}
	
	\paragraph{Computing Mode for Range Queries.} A mode of a multiset $S$ is an element $a \in S$ of maximum multiplicity; that
	is, $a$ occurs at least as frequently as any other element in $S$. Given a sequence $\mT$
	of $n$ elements and set of $m$ range queries $\Q$, for each query $q\in \Q$, the goal is to answer the mode of $\mT\cap q$.

	\begin{theorem}[Chan et al.\ \cite{chan2014linear}] \label{thm:chan}
		Suppose there is an algorithm that takes as input a sequence  
		of $n$ elements and set of $m$ range queries, runs in time $O(n^{c}\cdot (m+n))$ (for some $c\ge 0)$, and outputs the mode for all range queries, then   Boolean matrix multiplication on two $\sqrt{n}\times\sqrt{n}$ matrices can be solved
		in time $O(n^{1+c})$.
	\end{theorem}

	We are now ready to prove Theorem~\ref{thm:colorLB}.
	
	\begin{proof}[Proof of Theorem~\ref{thm:colorLB}]
		Suppose there is some $\varepsilon>0$, $C\in \mathbb{N}$, and a  combinatorial algorithm $\tilde{\mathcal{A}}$ to 
		answer the reporting version of the 
		$\CORLIS$ problem in $k/2+C\cdot n^{1/2-\varepsilon}\cdot (m+n)$ time. Then, we construct  an algorithm to refute \CBMMH in the following way. 
		
		Given as input a sequence $\mT:=(z_1,\ldots ,z_n)$
		of $n$ elements and set of $m$ mode range queries $\Q$, we construct an instance $(\mS:=(a_1,\ldots ,a_n),c:\mS\to [n],\Q)$ of $\CORLIS$ as follows. 
		Without loss of generality, we assume all elements in $\mT$ are in $[n]$. 
		For every $i\in [n]$, let $f_i$ be the number of times the $z_i$ has appeared in $\mT$ at an index less than $i$. Then, we define $a_i:=z_i+(f_i\cdot n)$. Note that we can compute $\mS$ from $\mT$ in $O(n\log n)$ time. Moreover, we define the color of $a_i$, i.e., $c(a_i)$ to simply be $z_i$. 
		
		We will now use $\tilde{\mathcal{A}}$ to answer the mode range queries in $\Q$   in $ 3C\cdot n^{1/2-\varepsilon}\cdot (m+n)$ time as follows (and this would contradict \CBMMH from Theorem~\ref{thm:chan}). 
		
		We feed $(\mS,c,\Q)$ to $\tilde{\mathcal{A}}$ and obtain in $k/2+C\cdot n^{1/2-\varepsilon}\cdot (m+n)+O(n\log n)$ time, for every $q\in \Q$, the longest monochromatic increase subsequence $\mS_q$ in $\mS\cap q$. The answer to the mode range query $q$ in $\mT$ is then simply the color of any element in $\mS_q$. Note that since the output is of size $k$ and we are running in $k/2+C\cdot n^{1/2-\varepsilon}\cdot (m+n)$ time (which includes the time needed to output), this means that $k\le 2C\cdot n^{1/2-\varepsilon}\cdot (m+n)$. Thus, the total runtime is $k/2+C\cdot n^{1/2-\varepsilon}\cdot (m+n)+O(n\log n)\le 3C\cdot n^{1/2-\varepsilon}\cdot (m+n)$ for large enough $n$.

		Note that the lower bound continues to hold  even when we are only required to report the color of the longest monochromatic increasing subsequence for each range query, as the color corresponds to the mode. 
	\end{proof}

	\section{Open Problems}\label{sec:open}
	We conclude the paper with a few
	open problems. 
	\begin{itemize}
		\item Is it possible to extend the algorithm in \cite{GGK24} for the $\ORLIS$ problem to solve the  $\TRLIS$ problem in  $O(n\text{ polylog } n + k)$ time? Or, is there 
		a (conditional) hardness result which makes obtaining
		the upper bound unlikely?
		\item Another interesting research direction is to design a deterministic algorithm
		which runs in sub-quadratic time. Specifically, can 
		the construction of 
		stitching set  be efficiently derandomized?
		\item For the $\CRLIS$ problem, can we bridge the gap between the
		upper bound and the (conditional) lower bound. 
		We conjecture that the upper bound can be further
		improved.
		\item Can we extend our algorithmic technique to beat the quadratic barrier for the weighted version of $\ORLIS$? 
		\item Finally, it would be interesting to explore the $\RLIS$ in the dynamic model. 
		
	\end{itemize}
	
	\subsection*{Acknowledgments}
	We thank the anonymous reviewers for several helpful comments that improved the presentation of our results. 
	
	\bibliographystyle{alpha}
	\bibliography{refs}

\newcommand{\etalchar}[1]{$^{#1}$}
\begin{thebibliography}{BFGLO09}

\bibitem[AAL09]{aal09}
Peyman Afshani, Lars Arge, and Kasper~Dalgaard Larsen.
\newblock Orthogonal range reporting in three and higher dimensions.
\newblock In {\em Proceedings of Annual {IEEE} Symposium on Foundations of
  Computer Science ({FOCS})}, pages 149--158, 2009.

\bibitem[ABZ11]{abz11}
Peyman Afshani, Gerth~Stolting Brodal, and Norbert Zeh.
\newblock Ordered and unordered top-k range reporting in large data sets.
\newblock In {\em Proceedings of the Annual {ACM-SIAM} Symposium on Discrete
  Algorithms ({SODA})}, pages 390--400, 2011.

\bibitem[AE98]{ae98}
Pankaj~K. Agarwal and Jeff Erickson.
\newblock Geometric range searching and its relatives.
\newblock {\em Advances in Discrete and Computational Geometry}, pages 1--56,
  1998.

\bibitem[AFK{\etalchar{+}}24]{abboud2024new}
Amir Abboud, Nick Fischer, Zander Kelley, Shachar Lovett, and Raghu Meka.
\newblock New graph decompositions and combinatorial boolean matrix
  multiplication algorithms.
\newblock In {\em {STOC}}, 2024.
\newblock To appear.

\bibitem[AGH{\etalchar{+}}04]{albert2004longest}
Michael~H Albert, Alexander Golynski, Ang{\`e}le~M Hamel, Alejandro
  L{\'o}pez-Ortiz, S~Srinivasa Rao, and Mohammad~Ali Safari.
\newblock Longest increasing subsequences in sliding windows.
\newblock {\em Theoretical Computer Science}, 321(2-3):405--414, 2004.

\bibitem[AGM{\etalchar{+}}90]{altschul1990basic}
Stephen~F Altschul, Warren Gish, Webb Miller, Eugene~W Myers, and David~J
  Lipman.
\newblock Basic local alignment search tool.
\newblock {\em Journal of molecular biology}, 215(3):403--410, 1990.

\bibitem[AGM02]{agm02}
Pankaj~K. Agarwal, Sathish Govindarajan, and S.~Muthukrishnan.
\newblock Range searching in categorical data: Colored range searching on grid.
\newblock In {\em Proceedings of European Symposium on Algorithms ({ESA})},
  pages 17--28, 2002.

\bibitem[AJKS02]{ajtai2002approximate}
Mikl{\'o}s Ajtai, TS~Jayram, Ravi Kumar, and D~Sivakumar.
\newblock Approximate counting of inversions in a data stream.
\newblock In {\em Proceedings of {ACM} Symposium on Theory of Computing
  ({STOC})}, pages 370--379, 2002.

\bibitem[ANSS22]{andoni2022estimating}
Alexandr Andoni, Negev~Shekel Nosatzki, Sandip Sinha, and Clifford Stein.
\newblock Estimating the longest increasing subsequence in nearly optimal time.
\newblock In {\em Proceedings of Annual {IEEE} Symposium on Foundations of
  Computer Science ({FOCS})}, pages 708--719, 2022.

\bibitem[Arg03]{a03}
Lars Arge.
\newblock The buffer tree: A technique for designing batched external data
  structures.
\newblock {\em Algorithmica}, 37(1):1--24, 2003.

\bibitem[AW14]{abboud2014popular}
Amir Abboud and Virginia~Vassilevska Williams.
\newblock Popular conjectures imply strong lower bounds for dynamic problems.
\newblock In {\em 2014 IEEE 55th Annual Symposium on Foundations of Computer
  Science}, pages 434--443. IEEE, 2014.

\bibitem[BCKO08]{bcko08}
Mark~de {Berg}, Otfried Cheong, Marc~van {{Kreveld}}, and Mark Overmars.
\newblock {\em Computational Geometry: Algorithms and Applications}.
\newblock Springer-Verlag, 3rd edition, 2008.

\bibitem[BDHS13]{ballard2013graph}
Grey Ballard, James Demmel, Olga Holtz, and Oded Schwartz.
\newblock Graph expansion and communication costs of fast matrix
  multiplication.
\newblock {\em Journal of the ACM (JACM)}, 59(6):1--23, 2013.

\bibitem[BDR10]{bdr10}
Gerth~St{\o}lting Brodal, Pooya Davoodi, and S.~Srinivasa Rao.
\newblock On space efficient two dimensional range minimum data structures.
\newblock In {\em Proceedings of European Symposium on Algorithms ({ESA})},
  pages 171--182, 2010.

\bibitem[BFGLO09]{bfgl09}
Gerth~St{\o}lting Brodal, Rolf Fagerberg, Mark Greve, and Alejandro
  Lopez-Ortiz.
\newblock Online sorted range reporting.
\newblock In {\em International Symposium on Algorithms and Computation
  ({ISAAC})}, pages 173--182, 2009.

\bibitem[BGJS11]{bgjs11}
Gerth~St{\o}lting Brodal, Beat Gfeller, Allan~Gronlund Jorgensen, and Peter
  Sanders.
\newblock Towards optimal range medians.
\newblock {\em {Theoretical Computer Science}}, 412(24):2588--2601, 2011.

\bibitem[BKMT95]{bkmt95}
Panayiotis Bozanis, Nectarios Kitsios, Christos Makris, and Athanasios~K.
  Tsakalidis.
\newblock New upper bounds for generalized intersection searching problems.
\newblock In {\em Proceedings of International Colloquium on Automata,
  Languages and Programming ({ICALP})}, pages 464--474, 1995.

\bibitem[BL14]{bl14}
Gerth~St{\o}lting Brodal and Kasper~Green Larsen.
\newblock Optimal planar orthogonal skyline counting queries.
\newblock In {\em Scandinavian Symposium and Workshops on Algorithm Theory
  (SWAT)}, pages 110--121, 2014.

\bibitem[Bro16]{b16}
Gerth~Stolting Brodal.
\newblock External memory three-sided range reporting and top-$k$ queries with
  sublogarithmic updates.
\newblock In {\em Proceedings of Symposium on Theoretical Aspects of Computer
  Science ({STACS})}, volume~47, pages 23:1--23:14. 2016.

\bibitem[BT11]{bt11}
Gerth~Stolting Brodal and Konstantinos Tsakalidis.
\newblock Dynamic planar range maxima queries.
\newblock In {\em Proceedings of International Colloquium on Automata,
  Languages and Programming ({ICALP})}, pages 256--267, 2011.

\bibitem[CDL{\etalchar{+}}14]{chan2014linear}
Timothy~M Chan, Stephane Durocher, Kasper~Green Larsen, Jason Morrison, and
  Bryan~T Wilkinson.
\newblock Linear-space data structures for range mode query in arrays.
\newblock {\em Theory of Computing Systems}, 55(4):719--741, 2014.

\bibitem[CH21]{ch21}
Timothy~M. Chan and Zhengcheng Huang.
\newblock Dynamic colored orthogonal range searching.
\newblock In {\em Proceedings of European Symposium on Algorithms ({ESA})},
  volume 204, pages 28:1--28:13, 2021.

\bibitem[CHN20]{chn20}
Timothy~M. Chan, Qizheng He, and Yakov Nekrich.
\newblock Further results on colored range searching.
\newblock In {\em International Symposium on Computational Geometry (SoCG)},
  pages 28:1--28:15, 2020.

\bibitem[CLP11]{clp11}
Timothy~M. Chan, Kasper~Green Larsen, and Mihai Patrascu.
\newblock Orthogonal range searching on the ram, revisited.
\newblock In {\em Proceedings of Symposium on Computational Geometry
  ({S}o{CG})}, pages 1--10, 2011.

\bibitem[CN20]{cn20}
Timothy~M. Chan and Yakov Nekrich.
\newblock Better data structures for colored orthogonal range reporting.
\newblock In {\em Proceedings of the Annual {ACM-SIAM} Symposium on Discrete
  Algorithms ({SODA})}, pages 627--636, 2020.

\bibitem[EJ08]{ergun2008distance}
Funda Ergun and Hossein Jowhari.
\newblock On distance to monotonicity and longest increasing subsequence of a
  data stream.
\newblock In {\em Proceedings of the Annual {ACM-SIAM} Symposium on Discrete
  Algorithms ({SODA})}, pages 730--736, 2008.

\bibitem[Fre75]{fredman1975computing}
Michael~L Fredman.
\newblock On computing the length of longest increasing subsequences.
\newblock {\em Discrete Mathematics}, 11(1):29--35, 1975.

\bibitem[GG10]{gal2010lower}
Anna G{\'a}l and Parikshit Gopalan.
\newblock Lower bounds on streaming algorithms for approximating the length of
  the longest increasing subsequence.
\newblock {\em SIAM Journal on Computing}, 39(8):3463--3479, 2010.

\bibitem[GGK24]{GGK24}
Pawel Gawrychowski, Egor Gorbachev, and Tomasz Kociumaka.
\newblock Core-sparse monge matrix multiplication: Improved algorithm and
  applications.
\newblock {\em CoRR}, abs/2408.04613, 2024.

\bibitem[GJ21a]{gawrychowski2021conditional}
Pawe{\l} Gawrychowski and Wojciech Janczewski.
\newblock Conditional lower bounds for variants of dynamic {LIS}.
\newblock {\em arXiv preprint arXiv:2102.11797}, 2021.

\bibitem[GJ21b]{gj21}
Pawel Gawrychowski and Wojciech Janczewski.
\newblock Fully dynamic approximation of {LIS} in polylogarithmic time.
\newblock In {\em Proceedings of {ACM} Symposium on Theory of Computing
  ({STOC})}, pages 654--667, 2021.

\bibitem[GJKK07]{GopalanJKK07}
Parikshit Gopalan, T.~S. Jayram, Robert Krauthgamer, and Ravi Kumar.
\newblock Estimating the sortedness of a data stream.
\newblock In {\em Proceedings of the Annual {ACM-SIAM} Symposium on Discrete
  Algorithms ({SODA})}, pages 318--327, 2007.

\bibitem[GJRS18]{gjrs18}
Prosenjit Gupta, Ravi Janardan, Saladi Rahul, and Michiel H.~M. Smid.
\newblock Computational geometry: Generalized (or colored) intersection
  searching.
\newblock {\em In Handbook of Data Structures and Applications, CRC Press, 2nd
  edition}, page 1042–1057, 2018.

\bibitem[GJS95]{gjs95}
Prosenjit Gupta, Ravi Janardan, and Michiel H.~M. Smid.
\newblock Further results on generalized intersection searching problems:
  Counting, reporting, and dynamization.
\newblock {\em Journal of Algorithms}, 19(2):282--317, 1995.

\bibitem[HKNS15]{henzinger2015unifying}
Monika Henzinger, Sebastian Krinninger, Danupon Nanongkai, and Thatchaphol
  Saranurak.
\newblock Unifying and strengthening hardness for dynamic problems via the
  online matrix-vector multiplication conjecture.
\newblock In {\em Proceedings of the forty-seventh annual ACM symposium on
  Theory of computing}, pages 21--30, 2015.

\bibitem[HQT14]{hqt14}
Xiaocheng Hu, Miao Qiao, and Yufei Tao.
\newblock Independent range sampling.
\newblock In {\em Proceedings of ACM Symposium on Principles of Database
  Systems ({PODS})}, pages 246--255, 2014.

\bibitem[JL93]{jl93}
Ravi Janardan and Mario~A. Lopez.
\newblock Generalized intersection searching problems.
\newblock {\em International Journal of Computational Geometry and
  Applications}, 3(1):39--69, 1993.

\bibitem[JMA07]{jin2007trend}
Ruoming Jin, Scott McCallen, and Eivind Almaas.
\newblock Trend motif: A graph mining approach for analysis of dynamic complex
  networks.
\newblock In {\em Proceedings of International Conference on Management of Data
  ({ICDM})}, pages 541--546, 2007.

\bibitem[KMS05]{kms05}
Danny Krizanc, Pat Morin, and Michiel H.~M. Smid.
\newblock Range mode and range median queries on lists and trees.
\newblock {\em Nordic Journal of Computing}, 12(1):1--17, 2005.

\bibitem[KN11]{kn11}
Marek Karpinski and Yakov Nekrich.
\newblock Top-k color queries for document retrieval.
\newblock In {\em Proceedings of the Annual {ACM-SIAM} Symposium on Discrete
  Algorithms ({SODA})}, pages 401--411, 2011.

\bibitem[KOO{\etalchar{+}}18]{kiyomi2018space}
Masashi Kiyomi, Hirotaka Ono, Yota Otachi, Pascal Schweitzer, and Jun Tarui.
\newblock Space-efficient algorithms for longest increasing subsequence.
\newblock In {\em Proceedings of Symposium on Theoretical Aspects of Computer
  Science ({STACS})}, 2018.

\bibitem[KRSV07]{krsv07}
Haim Kaplan, Natan Rubin, Micha Sharir, and Elad Verbin.
\newblock Counting colors in boxes.
\newblock In {\em Proceedings of the Annual {ACM-SIAM} Symposium on Discrete
  Algorithms ({SODA})}, pages 785--794, 2007.

\bibitem[KS21]{ks21}
Tomasz Kociumaka and Saeed Seddighin.
\newblock Improved dynamic algorithms for longest increasing subsequence.
\newblock In {\em Proceedings of {ACM} Symposium on Theory of Computing
  ({STOC})}, pages 640--653, 2021.

\bibitem[KSV06]{ksv06}
Haim Kaplan, Micha Sharir, and Elad Verbin.
\newblock Colored intersection searching via sparse rectangular matrix
  multiplication.
\newblock In {\em Proceedings of Symposium on Computational Geometry
  ({S}o{CG})}, pages 52--60, 2006.

\bibitem[Lee02]{lee2002fast}
Lillian Lee.
\newblock Fast context-free grammar parsing requires fast boolean matrix
  multiplication.
\newblock {\em Journal of the ACM (JACM)}, 49(1):1--15, 2002.

\bibitem[LNVZ06]{liben2006finding}
David Liben-Nowell, Erik Vee, and An~Zhu.
\newblock Finding longest increasing and common subsequences in streaming data.
\newblock {\em Journal of Combinatorial Optimization}, 11:155--175, 2006.

\bibitem[LP12]{lp12}
Kasper~Green Larsen and Rasmus Pagh.
\newblock {I/O}-efficient data structures for colored range and prefix
  reporting.
\newblock In {\em Proceedings of the Annual {ACM-SIAM} Symposium on Discrete
  Algorithms ({SODA})}, pages 583--592, 2012.

\bibitem[LPS08]{lps08}
Ying~Kit Lai, Chung~Keung Poon, and Benyun Shi.
\newblock Approximate colored range and point enclosure queries.
\newblock {\em Journal of Discrete Algorithms}, 6(3):420--432, 2008.

\bibitem[LvW13]{lw13}
Kasper~Green Larsen and Freek van Walderveen.
\newblock Near-optimal range reporting structures for categorical data.
\newblock In {\em Proceedings of the Annual {ACM-SIAM} Symposium on Discrete
  Algorithms ({SODA})}, pages 265--276, 2013.

\bibitem[LZZZ17]{li2017longest}
Youhuan Li, Lei Zou, Huaming Zhang, and Dongyan Zhao.
\newblock Longest increasing subsequence computation over streaming sequences.
\newblock {\em IEEE Transactions on Knowledge and Data Engineering ({TKDE})},
  30(6):1036--1049, 2017.

\bibitem[Mal62]{mallows1962patience}
Colin~L Mallows.
\newblock Patience sorting.
\newblock {\em SIAM Review}, 4(2):148--149, 1962.

\bibitem[Mal63]{mallows1963patience}
Colin~L Mallows.
\newblock Patience sorting.
\newblock {\em SIAM Review}, 5(4):375, 1963.

\bibitem[MS20]{mitzenmacher2020dynamic}
Michael Mitzenmacher and Saeed Seddighin.
\newblock Dynamic algorithms for {LIS} and distance to monotonicity.
\newblock In {\em Proceedings of {ACM} Symposium on Theory of Computing
  ({STOC})}, pages 671--684, 2020.

\bibitem[Mut02]{m02}
S.~Muthukrishnan.
\newblock Efficient algorithms for document retrieval problems.
\newblock In {\em Proceedings of the Annual {ACM-SIAM} Symposium on Discrete
  Algorithms ({SODA})}, pages 657--666, 2002.

\bibitem[Nek14]{n14}
Yakov Nekrich.
\newblock Efficient range searching for categorical and plain data.
\newblock {\em ACM Transactions on Database Systems ({TODS})}, 39(1):9, 2014.

\bibitem[NR23]{nr23}
Yakov Nekrich and Saladi Rahul.
\newblock 4d range reporting in the pointer machine model in almost-optimal
  time.
\newblock In {\em Proceedings of the Annual {ACM-SIAM} Symposium on Discrete
  Algorithms ({SODA})}, pages 1862--1876, 2023.

\bibitem[NV13]{nv13}
Yakov Nekrich and Jeffrey~Scott Vitter.
\newblock Optimal color range reporting in one dimension.
\newblock In {\em Proceedings of European Symposium on Algorithms ({ESA})},
  pages 743--754, 2013.

\bibitem[NV21]{newman2021}
Ilan Newman and Nithin Varma.
\newblock New sublinear algorithms and lower bounds for {LIS} estimation.
\newblock In {\em Proceedings of International Colloquium on Automata,
  Languages and Programming ({ICALP})}, pages 100:1--100:20, 2021.

\bibitem[PTS{\etalchar{+}}14]{ptsnv14}
Manish Patil, Sharma~V. Thankachan, Rahul Shah, Yakov Nekrich, and
  Jeffrey~Scott Vitter.
\newblock Categorical range maxima queries.
\newblock In {\em Proceedings of ACM Symposium on Principles of Database
  Systems ({PODS})}, pages 266--277, 2014.

\bibitem[Rah17]{r17}
Saladi Rahul.
\newblock Approximate range counting revisited.
\newblock In {\em 33rd International Symposium on Computational Geometry
  (SoCG)}, volume~77, pages 55:1--55:15, 2017.

\bibitem[Rah21]{r21}
Saladi Rahul.
\newblock Approximate range counting revisited.
\newblock {\em Journal of Computational Geometry}, 12(1):40--69, 2021.

\bibitem[RJ12]{rj12}
Saladi Rahul and Ravi Janardan.
\newblock Algorithms for range-skyline queries.
\newblock In {\em Proceedings of ACM Symposium on Advances in Geographic
  Information Systems ({GIS})}, pages 526--529, 2012.

\bibitem[RSSS19]{rubinstein2019approximation}
Aviad Rubinstein, Saeed Seddighin, Zhao Song, and Xiaorui Sun.
\newblock Approximation algorithms for lcs and {LIS} with truly improved
  running times.
\newblock In {\em Proceedings of Annual {IEEE} Symposium on Foundations of
  Computer Science ({FOCS})}, pages 1121--1145, 2019.

\bibitem[RT15]{rt15}
Saladi Rahul and Yufei Tao.
\newblock On top-k range reporting in 2d space.
\newblock In {\em Proceedings of ACM Symposium on Principles of Database
  Systems ({PODS})}, pages 265--275, 2015.

\bibitem[RT16]{rt16}
Saladi Rahul and Yufei Tao.
\newblock Efficient top-k indexing via general reductions.
\newblock In {\em Proceedings of ACM Symposium on Principles of Database
  Systems ({PODS})}, pages 277--288, 2016.

\bibitem[RT19]{rt19}
Saladi Rahul and Yufei Tao.
\newblock A guide to designing top-k indexes.
\newblock {\em {SIGMOD} Record}, 48(2):6--17, 2019.

\bibitem[RZ11]{roditty2011dynamic}
Liam Roditty and Uri Zwick.
\newblock On dynamic shortest paths problems.
\newblock {\em Algorithmica}, 61:389--401, 2011.

\bibitem[Sat94]{satta1994tree}
Giorgio Satta.
\newblock Tree-adjoining grammar parsing and boolean matrix multiplication.
\newblock {\em Computational linguistics}, 20(2):173--191, 1994.

\bibitem[SJ05]{sj05}
Qingmin Shi and Joseph J{\'a}J{\'a}.
\newblock Optimal and near-optimal algorithms for generalized intersection
  reporting on pointer machines.
\newblock {\em Information Processing Letters ({IPL})}, 95(3):382--388, 2005.

\bibitem[SS17]{saks2017estimating}
Michael Saks and C~Seshadhri.
\newblock Estimating the longest increasing sequence in polylogarithmic time.
\newblock {\em SIAM Journal of Computing}, 46(2):774--823, 2017.

\bibitem[ST86]{st86}
Neil Sarnak and Robert~Endre Tarjan.
\newblock Planar point location using persistent search trees.
\newblock {\em Communications of the ACM ({CACM})}, 29(7):669--679, 1986.

\bibitem[SW07]{sun2007communication}
Xiaoming Sun and David~P Woodruff.
\newblock The communication and streaming complexity of computing the longest
  common and increasing subsequences.
\newblock In {\em Proceedings of the Annual {ACM-SIAM} Symposium on Discrete
  Algorithms ({SODA})}, pages 336--345, 2007.

\bibitem[Tao22]{t22}
Yufei Tao.
\newblock Algorithmic techniques for independent query sampling.
\newblock In {\em Proceedings of ACM Symposium on Principles of Database
  Systems ({PODS})}, pages 129--138, 2022.

\bibitem[Tis08a]{Tiskin08a}
Alexander Tiskin.
\newblock Semi-local longest common subsequences in subquadratic time.
\newblock {\em Journal of Discrete Algorithms}, 6(4):570--581, 2008.

\bibitem[Tis08b]{Tiskin08b}
Alexandre Tiskin.
\newblock Semi-local string comparison: Algorithmic techniques and
  applications.
\newblock {\em Math. Comput. Sci.}, 1(4):571--603, 2008.

\bibitem[Tis10]{Tiskin10}
Alexander Tiskin.
\newblock Fast distance multiplication of unit-monge matrices.
\newblock In {\em Proceedings of the Annual {ACM-SIAM} Symposium on Discrete
  Algorithms ({SODA})}, pages 1287--1296, 2010.

\bibitem[XLRJ20]{xlrj20}
Jie Xue, Yuan Li, Saladi Rahul, and Ravi Janardan.
\newblock Searching for the closest-pair in a query translate.
\newblock {\em Journal of Computational Geometry}, 11(2):26--61, 2020.

\bibitem[XLRJ22]{xlrj22}
Jie Xue, Yuan Li, Saladi Rahul, and Ravi Janardan.
\newblock New bounds for range closest-pair problems.
\newblock {\em Discrete {\&} Computational Geometry}, 68(1):1--49, 2022.

\bibitem[Xue19]{xue19}
Jie Xue.
\newblock Colored range closest-pair problem under general distance functions.
\newblock In {\em Proceedings of the Annual {ACM-SIAM} Symposium on Discrete
  Algorithms ({SODA})}, pages 373--390, 2019.

\bibitem[Zha03]{zhang2003alignment}
Hongyu Zhang.
\newblock Alignment of blast high-scoring segment pairs based on the longest
  increasing subsequence algorithm.
\newblock {\em Bioinformatics}, 19(11):1391--1396, 2003.

\end{thebibliography}
	
\end{document}